\documentclass[journal,twoside,web]{ieeecolor}
\usepackage{include_packages}
\hypersetup{hidelinks=true}
\usepackage{textcomp}
\def\BibTeX{{\rm B\kern-.05em{\sc i\kern-.025em b}\kern-.08em
    T\kern-.1667em\lower.7ex\hbox{E}\kern-.125emX}}
\title{Collaborative Decision-Making and the $k$-Strong Price of Anarchy in Common Interest Games}
\author{Bryce L. Ferguson, \IEEEmembership{Student Member, IEEE}, Dario Paccagnan, \IEEEmembership{Member, IEEE}, Bary S. R. Pradelski, and Jason R. Marden, \IEEEmembership{Senior Member, IEEE}
\thanks{This work was supported in part by the Office of Naval Research under Grant \# N00014-20-1-2359, the Air Force Office of Scientific Research under Grants \# FA95550-20-1-0054 and \# FA9550-21-1-0203, and the
French National Research Agency (ANR) under Grant \# ANR-19-CE48-0018-01.}
\thanks{Bryce L. Ferguson and Jason R. Marden are with the Department of Electrical and Computer Engineering at the University of California, Santa Barbara, CA. \{\texttt{blferguson,jrmarden}\}\texttt{@ece.ucsb.edu}.}
\thanks{Dario Paccagnan is with the Department of Computing at Imperial College London, UK. \texttt{d.paccagnan@imperial.ac.uk}.}
\thanks{Bary S. R. Pradelski is with the National Centre for Scientific Research (CNRS) and the Department of Economics at Oxford, UK.\texttt{bary.pradelski@cnrs.fr}.}}
\begin{document}

\maketitle

\begin{abstract}
	The control of large-scale, multi-agent systems often entails distributing decision-making across the system components.
	However, with advances in communication and computation technologies, we can consider new collaborative decision-making paradigms that bridge centralized and distributed control architectures.
	In this work, we seek to understand the benefits and costs of increased collaborative communication in multi-agent systems.
	We specifically study this in the context of common interest games in which groups of up to $k$ agents can coordinate their actions in maximizing a common objective function.
	The equilibria that emerge in these systems are the \kstrong Nash equilibria of the common interest game; studying the properties of these states provides relevant insights into the efficacy of inter-agent collaboration.
	Our contributions come threefold: 1) provide bounds on how well \kstrong Nash equilibria approximate the optimal system welfare, formalized by the \kstrong price of anarchy, 2) prove the run-time and transient performance of collaborative agent-based dynamics, and 3) introduce techniques of redesigning objectives for groups of agents which improve system performance.
	We study these three facets generally as well as in the context of resource allocation problems, in which we provide tractable linear programs that give tight bounds on the \kstrong price of anarchy.
\end{abstract}

\section{Introduction}

Large-scale systems such as transportation services \cite{Wollenstein2020}, robotic fleets \cite{khamis2015multi}, supply chains \cite{ranganathanReInventingFoodSupply2022}, or cloud computing services \cite{Tsakalozos2011} can be challenging to design effective control schemes for due to their many components and vast scale.
The two prevailing paradigms to design control schemes are centralized control~\cite{filip2008DecisionSupportControl,daini2022CentralizedTrafficControl,fang2015CentralizedResourceAllocation}, which guides behavior across the entire system and distributed control ~\cite{antonelli2013interconnected,mardenCooperativeControlPotential2009,murray2007recent}, which allows local components to guide their own behavior.
Each of these approaches possesses respective pros and cons: centralization allows for more direct manipulation of system behavior at the cost of greater communication and computation requirements, while decentralization reduces the communication and computation requirements but cannot always attain the desired system behavior.
Advancements in embedded communication and computation~\cite{das2019tarmac,Ferguson2022_Cost,siljak2005ControlLargescaleSystems,xu2022ResourceAwareDistributedSubmodular} enable the design of new paradigms that \textit{exist between centralized and distributed control.}

Specifically, we study the efficacy of 
learning in multi-agent systems when individual system components (or agents) can partially communicate and thus coordinate their behavior.
%Specifically, we study the efficacy of communication-enabled collaboration in multi-agent systems whereby individual system components (or agents) can partially coordinate their behavior with one another.
Many engineering domains are on the precipice of enabling these collaborative paradigms; for example, 
autonomous vehicle platoons with connected cruise control~\cite{Orosz2016},
unmanned aerial surveillance vehicles with range-limited communication \cite{nawaz2021UAVCommunicationNetworks}, and
cloud computing networks with emerging distributed learning techniques \cite{lazaridou2020emergent}.
In each of these settings, inter-agent communication and collaboration offer the opportunity to improve the performance attainable by the system as a whole; however, implementing these frameworks incurs costs that are both monetary--in the form of the additional technology required--and computational--in the form of more complex decision-making algorithms.
In this work, we provide tools to help better understand the \textit{benefits and costs associated with collaborative communication} in multi-agent systems.

\begin{figure*}[t]
    \centering
    \includegraphics[width=0.825\textwidth]{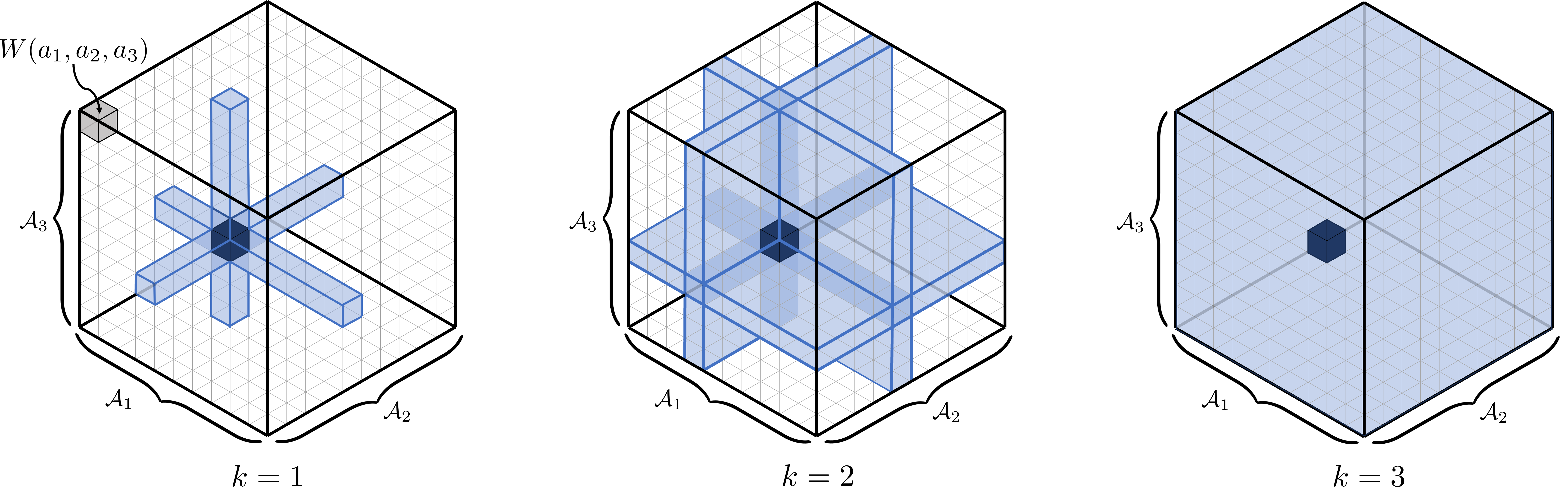}
    \caption{Illustration of the \kstrong Nash equilibrium local optimality guarantee for a three-agent common-interest game where $k \in \{1,2,3\}$.
    In each case, if the dark cube is a \kstrong Nash equilibrium, then it is optimal over the highlighted region with respect to the shared objective function $W$.
    As $k$ (the size of collaborative groups) increases, the local optimality is strengthened by holding overall $k$-lateral deviations.}
    \label{fig:cube_eg}
\end{figure*}

We model a multi-agent system as a common interest game where some (but not all) groups of agents can collaborate in selecting their actions to maximize the system welfare.
We particularly focus on the case where a collaborative action takes the form of a group best response, i.e., a group of agents updating their actions in response to the remaining players' actions.
%This group decision making adapts distributed setting...
As the size and number of these collaborative groups increase, a coordinated group decision has a larger impact on system behavior.
To range the level of collaboration between the fully distributed setting (where no agents can collaborate) and the fully centralized setting (where all agents can collaborate collectively), we consider the cases where groups of up to $k$ agents can collaborate.
In these collaborative environments, a stable state of the system is that of the \kstrong Nash equilibrium~\cite{aumann1959acceptable}.
Researchers have studied the existence \cite{nessah2014existence} and computation \cite{gatti2017verification} of strong Nash equilibria in settings including congestion games \cite{holzman1997strong}, lexicographical games \cite{harks2009strong}, and Markov games \cite{clempner2020finding}.
This work applies these concepts to multi-agent systems.
%We study agent-based dynamics which converge to \kstrong Nash equilibria
To understand the possible benefits of collaboration to system performance, we quantify how well \kstrong Nash equilibria approximate the optimal welfare, termed the \kstrong price of anarchy~\cite{epsteinStrongEquilibriumCost2007,andelmanStrongPriceAnarchy2009}.
To understand the possible cost of collaboration, we analyze the running time and transient performance of agent-based dynamics, which converge to \kstrong Nash equilibria.

Distributed learning in games has been a widely studied area in controls~\cite{Barreiro-Gomez2016}, but the ability to reach equilibrium with coalitional best responses has not yet been studied; we thus study the added run time of collaborative algorithms.
Quantifying the \kstrong price of anarchy has been studied in network formation games~\cite{epsteinStrongEquilibriumCost2007,andelmanStrongPriceAnarchy2009} and load balancing games \cite{fiat2007strong,epstein2012price,chien2009strong}, as well as more general utility maximizing games~\cite{bachrachStrongPriceAnarchy2014,feldman2015unified}.
In many of these, the bounds are either not tight (particularly for finitely many players) or hold for equilibria which need not exist.
By focusing on the class of common interest games, we guarantee the existence of collaborative equilibria, provide tight approximation bounds, and develop new insights into collaborative multi-agent optimization.

\textbf{\textit{Organization -}} This work provides tools to understand the benefits and costs of collaborative communication by studying the qualities of \kstrong Nash equilibria.
In \cref{subsec:smooth}, we consider the case where groups of agents are designed to maximize the system welfare and introduce the notion of \lmksmooth~games (a generalization of smooth games~\cite{Roughgarden2012} and coalitionally smooth games~\cite{bachrachStrongPriceAnarchy2014}), and provide bounds on the \kstrong price of anarchy.
Then, in \cref{subsec:ra}, we focus on the well-studied setting of distributed resource allocation problems~\cite{Gairing2009,Bilo2016,Paccagnan2018a,zhangMarkovGamesDecoupled2023, kondaOptimalDesignBest2022,Ferguson2022robust}, and provide tight bounds on the \kstrong price of anarchy via the solution of a tractable linear program.
\cref{fig:spoa_plot_4} plots these bounds and demonstrates how increased collaboration improves efficiency guarantees in several classes of resource allocation problems.
In \cref{sec:dyn}, we consider the effects of group decision-making on agent-based dynamics; specifically, we show the added run-time complexity of coalitional round-robin dynamics and provide transient performance guarantees of asynchronous best response dynamics.
We support our findings with numerical examples.
%, which highlight that collaborative agent-based dynamics provide better performance but require more evaluations of the system's welfare.
In \cref{sec:design}, we consider that the system operator may be able to design the agents' objective separately from the system welfare; we provide a generalized technique for bounding the \kstrong price of anarchy in this setting.
In \cref{subsec:gra}, we again focus on the setting of resource allocation and provide two linear programs to lower and upper bound the attainable \kstrong price of anarchy guarantee via utility design.
%a tractable linear program which provides a high performing utility rule and a second to determine if it is optimal.

%\textbf{\textit{Organization - }} In this work...

\section{Preliminaries}
Throughout, we will denote $[n] = \{1,\ldots,n\}$.
We will regularly use the binomial coefficient $\binom{n}{k} = \frac{n!}{(n-k)!k!}$ in constructing optimization problems; we define this value as $0$ when $n < k$ for ease of notation.

\subsection{Collaborative Decision Making}\label{subsec:collab}
Consider a finite set of agents $N = \{1,\ldots,n\}$.
Each agent $i \in N$ selects an action $a_i$ from a finite action set $\mathcal{A}_i$.
When each agent selects an action, we will denote their joint action by the tuple $a = (a_1,\ldots,a_n) \in \mathcal{A} = \mathcal{A}_1\times\cdots\times\mathcal{A}_n$.
Let $G = (N,\mathcal{A})$ be a tuple encoding the components of the agent environment.
The system's performance is dictated by the agents' actions; as such, for each joint-action $a$, we assign a system welfare $W(a)$ where $W:\mathcal{A}\rightarrow \mathbb{R}_{\geq 0}$ is the system designer's objective function.
With this, we let the tuple $(G,W)$ denote a \emph{multi-agent system} (often referred to as a system), which defines the primitives of the system designer's problem of designing an effective control algorithm.

The system designer would like to configure the agents to reach a joint action that maximizes the system welfare, i.e.,
\begin{equation}
\Opt{a} \in \argmax_{a \in \mathcal{A}} W(a).
\end{equation}
Though this system state is ideal, it may be difficult to attain as 1) solving for the optimal allocation can be combinatorial and in some cases (including those from \cref{subsec:ra}) NP-hard~\cite{Gairing2009}, and 2) it requires a centralized authority to control all agents, which may be practically or logistically difficult.
To resolve this, we will consider that agents make decisions in a decentralized manner.

Fully distributing the decision-making involves designing each agent to update their action locally and has been widely studied and developed to guarantee reasonable system behavior \cite{antonelli2013interconnected}; however, fully distributing decision-making may often become unnecessary as emerging communication technologies enable \textit{collaborative inter-agent decision-making}\cite{das2019tarmac}.
To implement one such collaborative system architecture, a system operator must make two decisions: 1) which group of agents can collaborate on their decisions (possibly subject to some operational constraints), and 2) how the agents should collaborate on their decisions.
A natural choice for the latter is a group best response.
Let $\Gamma \subseteq N$ be a group of agents endowed with the ability to collaboratively select a \textit{group action} $a_\Gamma \in \mathcal{A}_{\Gamma} = {\prod}_{i \in \Gamma} \mathcal{A}_i$, which they select by maximizing the system welfare over their group action-set,
\begin{equation}\label{eq:gamma_welf_update}
    a_{\Gamma} \in \argmax_{a_\Gamma^\prime \in \mathcal{A}_\Gamma} W(a_\Gamma^\prime,a_{-\Gamma}),
\end{equation}
where $a_{-\Gamma}$ denotes the actions of the players $i \in N \setminus \Gamma$.
If there are multiple elements in the argmax, the group breaks them at random unless they can remain with their current action.

Intuitively, a group best responding and collaboratively maximizing the system welfare should lead to direct improvements to system performance; however, one can consider other group decision-making rules as well.
In particular, in \cref{sec:design}, we will consider that the system designer can design the agents' objective separately from the system objective as a means to further shape system behavior.
In either case, one would imagine that the greater the collaborative structure, the greater the impact on emergent behavior.

For the system operator's decision over which groups should collaborate, let $\mathcal{C} \subseteq 2^N$ denote the \emph{collaboration set}, or the set of groups of agents ($\Gamma \in \mathcal{C}$) able to collaborate their decisions.
These collaborations can overlap--where agents can partake in multiple, disparate collaborations--and vary in size.
For example, if agents send signals through a communication network~\cite{saad2009coalitional}, we will have \hbox{$\mathcal{C} = \{ (i,j) \in N^2 \mid (i,j) \in E\}$} where $E$ are the edges in a communication graph.
If agents are allowed to communicate with each other one at a time and make pairwise decisions~\cite{bayram2013multirobot}, then $\mathcal{C} = \{ (i,j) \in N^2\}$.
If agents can only communicate with others within a local proximity~\cite{cappello2021distributed}, then $\mathcal{C} = \{ \Gamma \subseteq N \mid \rho(i,j) \leq d~\forall i,j\in \Gamma\}$ where $\rho$ measures the distance between two agents and $d$ is a maximum communication range.
Once the system operator decides on the collaborative structure and the group decision-making protocol, the agents' decision-making process forms a collaborative multi-agent system, denoted by the tuple $(G,W,\mathcal{C})$.

As we vary the number and size of collaborative sets, we can consider control paradigms somewhere between centralized (i.e., $\{N\} \in \mathcal{C}$) and fully distributed (i.e., $\mathcal{C} = \left\{\{1\},\{2\},\ldots,\{n\}\right\}$.
This work seeks to understand the efficacy of different levels of communication/collaboration.
To more effectively quantify this, we consider a specific type of collaboration set in which we can range between the centralized and distributed extremes.

\subsection{k-Strong Nash Equilibria}
We consider the collaboration sets that contain groups of agents up to size $k$.
Let \hbox{$\mathcal{C}_k = \{\Gamma \subseteq N \mid |\Gamma| = k \}$} denote the subsets of exactly $k$ agents and $\mathcal{C}_{[k]} = \bigcup_{\zeta \in [k]} \mathcal{C}_\zeta$ be the subsets that contain at most $k$ agents.
When $k=1$, we recover the fully distributed setting, and when $k=n$, we recover the fully centralized setting.
As we vary $k$ between $1$ and $n$, we sweep through different levels of communication and collaboration.

\begin{figure}
    \centering
    \includegraphics[width=0.485\textwidth]{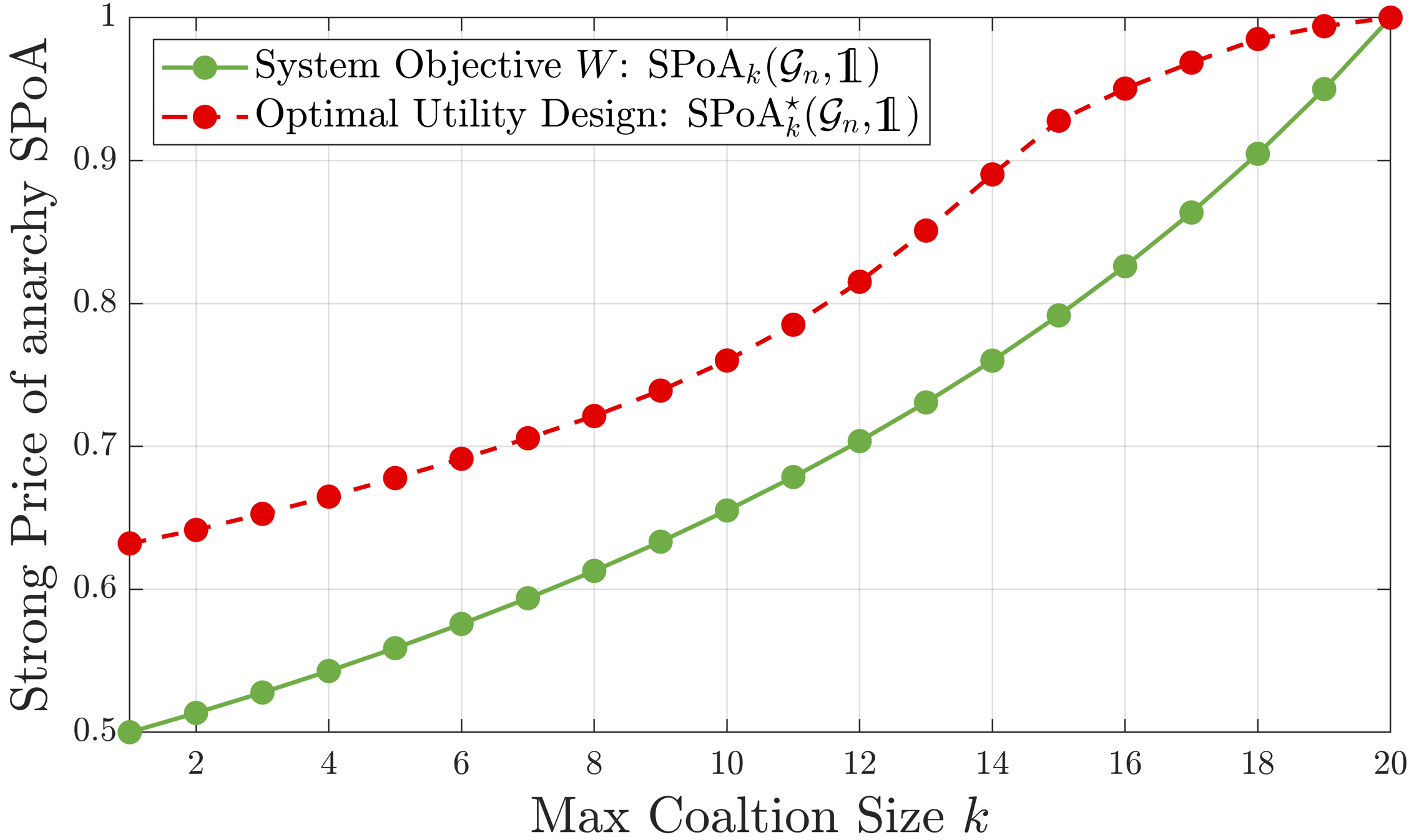}
    \caption{Strong Price of Anarchy in resource covering games with $n=20$ players and coalitions up to size $k$ (horizontal axis). As the size of groups that are allowed to collaborate grows, so too does the approximation ratio (i.e., strong price of anarchy) of a \kstrong Nash equilibrium. The efficiency of an equilibrium can be further improved by designing the utility functions agents are set to maximize. The solid green line is the \kstrong price of anarchy when agents maximize the system objective (generated by \cref{thm:spoa}). The dashed red line is an upper bound on the \kstrong price of anarchy while using an optimal utility design (generated by \cref{prop:util_upper_bound}).}
    \label{fig:spoa}
\end{figure}

In the game-theoretic approach to multi-agent systems, a Nash equilibrium is a joint action where no agent can unilaterally deviate their action to improve the system welfare~\cite{aumann1959acceptable}.
We generalize this concept to the setting of collaborative decision-making by considering a \emph{\kstrong Nash equilibrium} as a joint action where no group of $k$ agents can deviate their group's actions to improve the welfare.
\begin{definition}\label{def:ksnash}
A joint-action $\KSNash{a} \in \mathcal{A}$ is a \kstrong Nash equilibrium for the common-interest game $(G,W,\mathcal{C}_{[k]})$ if
\begin{equation}\label{eq:ksnash_def}
    W(\KSNash{a}) \geq W(a^\prime_\Gamma,\KSNash{a}_{-\Gamma}),~\forall a^\prime_\Gamma \in \mathcal{A}_\Gamma,~\Gamma \in \mathcal{C}_{[k]}.
\end{equation}
\end{definition}
Let $\KSNE(G,W) \subseteq \mathcal{A}$ denote the set of all \kstrong Nash equilibria.
Note that when $k=1$, we recover the classical definition of a Nash equilibrium, and when $k=n$, the equilibrium condition implies global optimality.
\cref{def:ksnash} differs slightly from the literature, where \kstrong Nash equilibria are defined by no group of agents deviating to a new group action that is Pareto-optimal for the group (i.e., no agent receives a lower payoff with respect to their individual utility function)~\cite{aumann1959acceptable}; when the agents respond to a common interest objective, the definitions are equivalent.
Additionally, in general games, \kstrong Nash equilibria need not exist; however, that is not the case in our setting due to the common-interest structure we impose on agent decision-making.
\begin{proposition}\label{prop:exist}
In a system $(G,W)$ with collaboration set $\mathcal{C}_{[k]}$ for any $k \in [n]$, a $k$-strong Nash equilibrium exists.
\end{proposition}
The proof appears in the appendix.

The main focus of this work is understanding how equilibrium performance changes with the level of collaborative communication.
Notice that \eqref{eq:ksnash_def} serves as a local optimality guarantee in the neighborhood of $k$-lateral deviations.
\cref{fig:cube_eg} depicts this for a three-player matrix game; when $k=1$, a $1$-strong Nash equilibrium (or just Nash equilibrium) is optimal over the unilateral deviations, when $k=2$ a $2$-strong Nash equilibrium is optimal over the bilateral deviations, and when $k=3=n$, the $3$-strong Nash equilibrium is optimal over the whole joint-action space.
From this, we observe that the local optimality guarantee is strengthened as we increase the level of collaboration $k$ (i.e., $k^\prime {\rm SNE}  \subseteq \KSNE$ for $k^\prime > k$).

To quantify the effect of varying $k$ on equilibrium performance, we consider the ratio of worst-case equilibrium welfare and the optimal attainable welfare, termed the \emph{\kstrong price of anarchy}.
\begin{equation}
    \spoa_k(G,W) = \frac{\min_{\KSNash{a} \in \KSNE(G,W)} W(\KSNash{a})}{\max_{\Opt{a} \in \mathcal{A}} W(\Opt{a}) } \in [0,1],
\end{equation}
where we let $0/0$ be defined as $1$ to ignore the degenerate case when no welfare is attainable.
In the multi-agent system $(G,W)$ with communication structure $\mathcal{C}_{[k]}$, every \kstrong Nash equilibrium approximates the optimal solution at least as well as $\spoa_k(G,W)$.
Accordingly, we will use the \kstrong price of anarchy to understand the efficiency associated with collaborative decision-making.
For example, in \cref{fig:spoa}, we depict the \kstrong price of anarchy in resource covering games~\cite{Gairing2009} for $1 \leq k \leq n$, illustrating the performance guarantees attainable between centralized and distributed control paradigms.
%as an approximation ratio of collaborative decision-making and study how this approximation changes with the level of collaborative communication.

\subsection{Summary of Contributions}
This work studies the benefits and costs of increased collaborative communication within multi-agent systems.
Our contributions come threefold:\\
1) In \cref{sec:quant}, we provide tools to quantify the \kstrong price of anarchy when agents optimize the system objective.
We introduce \lmksmooth~games and provide a \kstrong price of anarchy guarantee using the parameters $\lambda$ and $\mu$. We then focus on the class of resource allocation games, where in Proposition~\ref{prop:smoothLP}, we show that these parameters can be found via the solution to a tractable linear program. In \cref{thm:spoa}, we show that combining the constraints of each of the $k$ linear programs gives a tight bound. Figure \ref{fig:spoa_plot_4} depicts the \kstrong price of anarchy for several classes of resource allocation games.\\
2) In \cref{sec:dyn}, we study collaborative dynamics that reach these equilibria. In \cref{subsec:rr}, we introduce the coalitional round-robin dynamics and show that an equilibrium is reached in a finite number of best responses and that the number of welfare comparisons grows with a small-base exponential of $k$. In \cref{subsec:abr}, we introduce the asynchronous coalitional best response dynamics, which we show converge almost surely. Further, if the game is \lmksmooth, then we provide a bound on the transient performance (or the cumulative welfare along the dynamics). We support these findings with a numerical study in \cref{subsec:sim}.\\
3) In \cref{sec:design}, we consider how to improve the design of a group's decision-making process. By providing the agents with a new, designed objective function, the system designer may alter the set of equilibria and ideally increase the \kstrong price of anarchy. In \cref{subsec:gsmooth}, we generalize the notion of coalitional smoothness to the setting where the agents' objective differs from the system welfare, and in \cref{thm:gsmoothLP}, we show how we can construct an optimal utility rule. \cref{fig:spoa_opt_plot_4} shows the \kstrong price of anarchy under the optimal utility design for resource allocation games, demonstrating the added benefit of designing how groups of agents make decisions.

\begin{figure*}[t]
    \centering
    \includegraphics[width=0.975\textwidth]{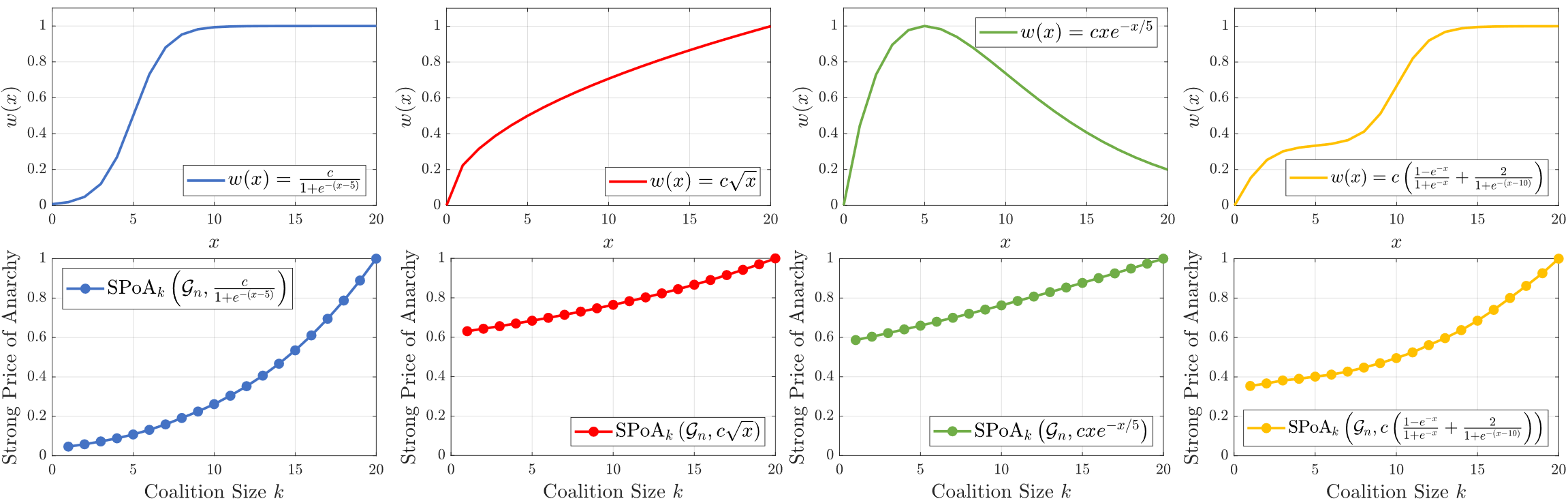}
    \caption{Tight \kstrong price of anarchy bounds for resource allocation games with various welfare functions. We illustrate four settings of local welfare function (top, left to right), and for each, we use \cref{thm:spoa} to generate tight bounds on the \kstrong price of anarchy for all $1 \leq k \leq n$. The bottom figures show these bounds and illustrate how increased inter-agent collaboration increases our efficiency guarantees on equilibrium system welfare.}
    \label{fig:spoa_plot_4}
\end{figure*}

\section{Quantifying $k$-Strong Price of Anarchy}\label{sec:quant}
\subsection{Coalitionally Smooth Games}\label{subsec:smooth}
We first consider the efficiency of \kstrong Nash equilibria for general multi-agent systems.
This efficiency--quantified by the \kstrong price of anarchy--is conditioned on the system welfare $W$ and the agent decision-making environment $G$.
In \cref{def:smooth}, we provide a condition on a system $(G,W)$ that will be useful in bounding the \kstrong price of anarchy.

\begin{definition}\label{def:smooth}
A system $(G,W)$ is \lmksmooth, where $\lambda,\mu \in \mathbb{R}_{\geq 0}^k$, if for all $a,a^\prime \in \mathcal{A}$
\begin{equation}\label{eq:smooth_def}
\frac{1}{\binom{n}{\zeta}} \sum_{\Gamma \in \mathcal{C}_\zeta} W(a_\Gamma^\prime,a_{-\Gamma}) \geq \lambda_\zeta W(a^\prime) - \mu_{\zeta}W(a),~\forall \zeta \in [k].
\end{equation}
\end{definition}
In \eqref{eq:smooth_def}, we provide a constraint on the welfare function stating that the average effect of a group of size $\zeta$ deviating their action from $a$ to $a^\prime$ is lower bounded by a linear combination of the welfare of $a$ and $a^\prime$.
The term smooth is in reference to the welfare function's change over the joint-action space being bounded by \eqref{eq:smooth_def}.
Additionally, \cref{def:smooth} extends the classic notion of smooth games~\cite{Roughgarden2012} and coalitional smoothness for strong equilibria~\cite{bachrachStrongPriceAnarchy2014} to the setting of $k$-coalitions in common interest games.

In effect, every system $(G,W)$ is smooth with $\lambda_\zeta = \mu_\zeta =0$ for all $\zeta \in [k]$, but some parameters $(\lambda,\mu)$ are more useful than others.
In \cref{prop:smooth}, we show that the parameters $\lambda$ and $\mu$ from \cref{def:smooth} can be used to lower bound the \kstrong price of anarchy.

\begin{proposition}\label{prop:smooth}
A system $(G,W)$ that is \lmksmooth~has $k$-strong price of anarchy satisfying
\begin{equation}\label{eq:smooth_prop_bound}
\spoa_k(G,W) \geq \frac{\lambda_\zeta}{1+\mu_\zeta},~\forall \zeta \in [k].
\end{equation}
\end{proposition}
\begin{proof}
Let $\KSNash{a} \in \mathcal{A}$ denote a \kstrong Nash equilibrium in the system $(G,W)$ (i.e., satisfying \cref{def:ksnash}), and let $\Opt{a} \in \argmax_{a \in \mathcal{A}} W(a)$ denote an optimal joint action.
For any $\zeta \in [k]$, we have
\begin{subequations}
\begin{align}
W(\KSNash{a}) &= \frac{1}{\binom{n}{\zeta}} \sum_{\Gamma \in \mathcal{C}_\zeta} W(\KSNash{a})\label{eq:smooth_proof_a}\\
&\geq \frac{1}{\binom{n}{\zeta}} \sum_{\Gamma \in \mathcal{C}_\zeta} W(\Opt{a}_\Gamma, \KSNash{a}_{-\Gamma})\label{eq:smooth_proof_b}\\
&\geq \lambda_\zeta W(\Opt{a}) - \mu_{\zeta}W(\KSNash{a}).\label{eq:smooth_proof_c}
\end{align}
\end{subequations}
Where \eqref{eq:smooth_proof_a} holds from $|\mathcal{C}_\zeta| = \binom{n}{\zeta}$, \eqref{eq:smooth_proof_b} holds from \cref{def:ksnash}, and \eqref{eq:smooth_proof_c} holds from \cref{def:smooth}.
Rearranging, we get $W(\KSNash{a})/W(\Opt{a}) \geq \lambda_\zeta/(1+\mu_\zeta)$.
\end{proof}
\eqref{eq:smooth_prop_bound} provides $k$ lower bounds on the \kstrong price of anarchy;
accordingly, a \kstrong Nash equilibrium approximates the system optimal at least as well as $\max_{\zeta \in [k]} \{\lambda_\zeta/(1+\mu_\zeta) \}$.
Often, the best lower bound is provided by $\zeta = k$; however, this is not true in general.
As such, we must consider each of the constraints in \eqref{eq:smooth_def} to derive the best bounds.

The efficiency bounds of this form are valuable for several reasons, including:
1) they can be used to provide insights on the transient guarantees of various multi-agent dynamics (see \cref{sec:dyn},
2) they easily generalize to broader equilibrium concepts (subject of future work), and
3) if parameters $(\lambda,\mu)$ can be shown to satisfy \eqref{eq:smooth_def} for a set of systems $\mathcal{S}$, then each system $(G,W) \in \mathcal{S}$ inherits the efficiency guarantee of \eqref{eq:smooth_prop_bound}.
This last point is particularly pertinent, as system models may be subject to noise, mischaracterizations, or changes over time.
If the efficiency guarantee holds across many similar systems, then the guarantees are essentially robust to these issues.

In the \cref{subsec:ra}, we will provide methods to find coalitional smoothness parameters for classes of resource allocation games via tractable linear programs.

\subsection{Resource Allocation Games}\label{subsec:ra}
In this subsection, we consider the well-studied class of resource allocation games~\cite{Paccagnan2018a,Gairing2009,zhangMarkovGamesDecoupled2023, kondaOptimalDesignBest2022,Vetta2002}.
Consider a set of resources or tasks $\mathcal{R} = \{1,\ldots,R\}$, to which agents are assigned, i.e., agent $i \in N$ selects a subset of these resources as its action $a_i \subseteq \mathcal{R}$ from a constrained set of subsets $\mathcal{A}_i \subseteq 2^\mathcal{R}$.
Each resource $r \in \mathcal{R}$ has a value $v_r \geq 0$; the welfare contributed by a resource is $v_r w(|a|_r)$, where $w: \{0,\ldots,n\} \rightarrow \mathbb{R}_{\geq 0}$ captures the added benefit of having multiple agents assigned to the same resources and $|a|_r$ is the number of agents assigned to $r$ in allocation $a$.
Assume that $w(0) = 0$ as no welfare is contributed by resources assigned to zero agents and further that $w(y) >0$ for all $y > 0$.
The system welfare is thus
\begin{equation}\label{eq:ra_welf}
    W(a) = \sum_{r \in \mathcal{R}} v_r w(|a|_r).
\end{equation}
For ease of notation, we will refer to the system welfare by the local welfare rule $w$, noting in the agent-environment $G$, it generates a welfare function $W$ via \eqref{eq:ra_welf}.

As discussed in \ref{subsec:smooth}, we wish to find efficiency bounds that hold over a class of resource allocation problems.
Let $G = (\mathcal{R},N,\mathcal{A},\{v_r\}_{r \in \mathcal{R}})$ denote a resource allocation problem, and let $\gee_n$ denote the set of all such resource allocation problems with at most $n$ agents.
In Proposition~\ref{prop:smoothLP}, we propose a tractable linear program whose solution provides parameters $(\lambda,\mu)$ which satisfy \cref{def:smooth} for every system $(G,w) \in \gee_n \times \{w\}$.
From \cref{prop:smooth}, this also provides a lower bound on the \kstrong price of anarchy for the class of resource allocation problems with local welfare $w$.
\begin{proposition}\label{prop:smoothLP}
Each resource allocation problem $(G,w) \in \mathcal{G}_n\times\{w\}$ is \lmksmooth~with $\lambda_\zeta = 1/\nu_\zeta^\star$ and $\mu_\zeta = \rho_\zeta^\star / \nu_\zeta^\star - 1$, where $(\rho_\zeta^\star,\nu_\zeta^\star)$ is a solution to the linear program \eqref{opt:smooth}:
\begin{align}
&(\rho_\zeta^\star,\nu_\zeta^\star)\in \argmin_{\rho\geq\nu \geq 0} ~~~ \rho \nonumber\\
&{\rm s.t.} \hspace{8pt}0 \geq w(o\hs +\hs x)-\rho w(e\hs +\hs x) + \nonumber\\
&~~\nu\hs \left( \hspace{-1.2pt} \hs \binom{n}{\zeta} w(e\hs +\hs x) \hs -\hs  \hs \sum_{\substack{0\leq\alpha\leq e\\ 0\leq\beta\leq o\\ \alpha+\beta \leq \zeta }} \hs {e\choose \alpha} {o \choose \beta}\binom{n\hs - \hs e\hs - \hs o}{\zeta \hs - \hs \alpha \hs - \hs \beta} w (e\hs +\hs x\hs +\hs \beta\hs -\hs \alpha)  \hs \right)\nonumber\\
&\hspace{160pt}\forall (e,x,o) \in \mathcal{I}\tag{P$\zeta$}\label{opt:smooth}
\end{align}
\end{proposition}

The constraints are parameterized by the triples $\mathcal{I} := \{(e,x,o)\in \mathbb{N}_{\geq 0}^3 \mid 1 \leq e+x+o \leq n\}$.
With the possibility of collaboration, an equilibrium becomes more difficult to characterize than in a fully distributed setting.
We circumvent this by introducing a parameterization which allows us to generalize the $\mathcal{O}\left(\sum_{\zeta = 1}^k{n \choose \zeta} m^\zeta\right)$ comparisons of \eqref{eq:ksnash_def} (where $m := \max_{i \in N} |\mathcal{A}_i|$) into $\mathcal{O}(n^3)$ linear inequalities.
Further, satisfying these inequalities provides parameters $(\lambda_\zeta,\mu_\zeta)$ that satisfy \cref{def:smooth}, leading to \eqref{opt:smooth} as a search for such parameters with the best \kstrong price of anarchy guarantee. 

\noindent\textit{Proof of Proposition~\ref{prop:smoothLP}}:
The proof largely relies on introducing a parameterization that lets us treat \eqref{eq:smooth_def} as a set of linear constraints.
Consider a resource allocation game $(G,w)\in \mathcal{G}_n\times\{w\}$ and any two actions $a,a^\prime \in \mathcal{A}$.
To each resource $r \in \mathcal{R}$, we assign a label $(e_r,x_r,o_r)$, where
\begin{align*}
    e_r &= \lvert\{ i \in N \mid r \in a_i \setminus a_i^\prime\}\rvert \\
    x_r &= \lvert\{ i \in N \mid r \in a_i \cap a_i^\prime\}\rvert \\
    o_r &= \lvert\{ i \in N \mid r \in a_i^\prime \setminus a_i\}\rvert.
\end{align*}
This is to say, $e_r$ denotes the number of agents utilizing resource $r$ in joint action $a$ but not $a^\prime$, $o_r$ is the number that uses resource $r$ in joint action $a^\prime$ but not $a$, and $x_r$ is the number that uses $r$ in both $a$ and $a^\prime$.
In the set of games $\gee_n$, let $\mathcal{I} = \{(e,x,o)\in \mathbb{N}_{\geq 0}^3 \mid 1 \leq e+x+o \leq n\}$ denote the set of possible labels, and
$\theta(e,x,o) := \sum_{r \in \mathcal{R}_{(e,x,o)}} v_r,$
where $\mathcal{R}_{(e,x,o)} = \{r \in \mathcal{R} \mid e_r = e, x_r = x, o_r=o\}$ denotes the set of resources with label $(e,x,o)$.
The parameter $\theta \in \mathbb{R}_{\geq 0}^{|\mathcal{I}|}$ is a vector with elements for each label.

We will now express the terms in \eqref{eq:smooth_def} using this parameterization.
Because $W(a) = \sum_{r \in \mathcal{R}} v_r w(|a|_r)$ depends only on the number of agents utilizing a resource; we can represent $|a|_r = e_r + x_r$ and write the system welfare as
\begin{align*}
    W(a) &= \sum_{r \in \mathcal{R}}v_r w(e_r+x_r)\\
    &=\sum_{(e,x,o) \in \mathcal{I}} \Bigg( \sum_{r \in \mathcal{R}_{e,x,o}} v_r \Bigg)w(e+x)\\
    &= \sum_{e,x,o}\theta(e,x,o)w(e+x).
\end{align*}
When not stated, the sum over $(e,x,o)$ is implied to be for each label in $\mathcal{I}$.
Similar steps can be followed to show $W(a^\prime) = \sum_{e,x,o} \theta(e,x,o)w(o+x)$.

Finally, the term $\sum_{\Gamma \in \mathcal{C}_\zeta} W(a_\Gamma^\prime,a_{-\Gamma})$ can similarly be transcribed by this parameterization:
\begin{align*}
&\sum_{\Gamma \in \mathcal{C}_\zeta}  W(a^\prime_\Gamma,a_{-\Gamma})\\
 &= \sum_{\Gamma\in\mathcal{C}_\zeta} \sum_{e,x,o} \sum_{r \in \mathcal{R}_{(e,x,o)}} \hspace{-2mm} v_r w(|a^\prime_\Gamma,a_{-\Gamma}|_r)\\
&= \sum_{e,x,o} \sum_{r \in \mathcal{R}_{(e,x,o)}}\hspace{-2mm} v_r \sum_{\Gamma\in\mathcal{C}_\zeta} w(|a^\prime_\Gamma,a_{-\Gamma}|_r)\\
&=  \sum_{e,x,o} \sum_{r \in \mathcal{R}_{(e,x,o)}}\hspace{-2mm} v_r \sum_{\substack{0 \leq \alpha \leq e\\ 0\leq \beta \leq o \\ \alpha+\beta \leq \zeta}} \hspace{-2mm} {e\choose \alpha} {o \choose \beta}\binom{n\hs - \hs e\hs - \hs o}{\zeta \hs - \hs \alpha \hs - \hs \beta}w(e \hspace{-1mm} + \hspace{-1mm} x \hspace{-1mm} + \hspace{-1mm} \beta \hspace{-1mm} - \hspace{-1mm} \alpha)\\
&= \sum_{e,x,o} \theta(e,x,o) \sum_{\substack{0 \leq \alpha \leq e\\ 0\leq \beta \leq o \\ \alpha+\beta \leq \zeta}} \hspace{-2mm} {e\choose \alpha} {o \choose \beta}\binom{n\hs - \hs e\hs - \hs o}{\zeta \hs - \hs \alpha \hs - \hs \beta}w(e \hspace{-1mm} + \hspace{-1mm} x \hspace{-1mm} + \hspace{-1mm} \beta \hspace{-1mm} - \hspace{-1mm} \alpha)
\end{align*}
where the set of coalitions $\mathcal{C}_\zeta$ was partitioned according to the action profile of the agents in each coalition.
We let $\alpha$ denote the number of agents in $\Gamma$ that utilize resource $r$ only in joint action $a$ and $\beta$ the number of agents in $\Gamma$ that utilize $r$ only in joint action $a^\prime$.
By simple counting arguments, there are exactly ${e\choose \alpha} {o \choose \beta}\binom{n -  e -  o}{\zeta  -  \alpha  -  \beta}$ coalitions grouped with the same $\alpha$ and $\beta$.
This decomposition is possible as the number of agents utilizing resource $r$ after a group $\Gamma$ deviates is precisely $e+x+\beta-\alpha$.

The smoothness constraint \eqref{eq:smooth_def} is satisfied only if
\begin{multline*}
\frac{1}{\binom{n}{\zeta}}\sum_{e,x,o} \theta(e,x,o) \sum_{\substack{0 \leq \alpha \leq e\\ 0\leq \beta \leq o \\ \alpha+\beta \leq \zeta}} \hspace{-2mm} {e\choose \alpha} {o \choose \beta}\binom{n\hs - \hs e\hs - \hs o}{\zeta \hs - \hs \alpha \hs - \hs \beta}w(e \hspace{-1mm} + \hspace{-1mm} x \hspace{-1mm} + \hspace{-1mm} \beta \hspace{-1mm} - \hspace{-1mm} \alpha)\\
 \geq \lambda_\zeta\sum_{e,x,o}\theta(e,x,o)w(o+x) -\mu_\zeta \sum_{e,x,o}\theta(e,x,o)w(e+x).
\end{multline*}
As $\theta(e,x,o) \geq 0$ for all $(e,x,o) \in \mathcal{I}$, it is sufficient to satisfy
\begin{multline}\label{eq:param_constraint}
\frac{1}{\binom{n}{\zeta}} \sum_{\substack{0 \leq \alpha \leq e\\ 0\leq \beta \leq o \\ \alpha+\beta \leq \zeta}} \hspace{-2mm} {e\choose \alpha} {o \choose \beta}\binom{n\hs - \hs e\hs - \hs o}{\zeta \hs - \hs \alpha \hs - \hs \beta}w(e \hspace{-1mm} + \hspace{-1mm} x \hspace{-1mm} + \hspace{-1mm} \beta \hspace{-1mm} - \hspace{-1mm} \alpha)\\
 \geq \lambda_\zeta w(o+x) -\mu_\zeta w(e+x),\quad \forall (e,x,o) \in \mathcal{I}.
\end{multline}
Observe that \eqref{eq:param_constraint} is independent of $a$, $a^\prime$, and $G$.
As such, this set of constraints serves as a sufficient condition that any $G \in \gee_n$ satisfies \eqref{eq:smooth_def} for all respective $a,a^\prime \in \mathcal{A}$.

To find parameters $\lambda_\zeta$ and $\mu_\zeta$ that provide the best \kstrong price of anarchy guarantee, we formulate the following optimization problem:
\begin{align}
\max_{\lambda_\zeta,\mu_\zeta \geq 0} \quad & \frac{\lambda_\zeta}{1+\mu_\zeta}\tag{P1$\zeta$}\label{opt:smooth_comp}\\
{\rm s.t.} \quad &~~ \eqref{eq:param_constraint}\nonumber
\end{align}
We restrict $\lambda_\zeta$ to be non-negative, though this constraint is not active except in degenerate cases.
Finally, we transform \eqref{opt:smooth_comp} by substituting new decision variables $\rho = (1+\mu_\zeta)/\lambda_\zeta$ and $\nu = 1/\left(\binom{n}{\zeta}\lambda_\zeta\right) \geq 0$.
The new objective becomes $1/\rho$.
Note that the constraint $(e,x,o)=(1,0,0)$ implies $\rho \geq 0$; we can thus invert the objective and change the minimization to a maximization, giving \eqref{opt:smooth}.
\hfill\qed

\begin{figure*}
\begin{align}
P^\star = &\min_{\rho,\{\nu_\zeta \geq 0\}_{\zeta \in [k]}} && \rho \nonumber\\
&~~~~~~{\rm s.t.} && \hspace{0pt}0 \geq w(o\hs +\hs x)-\rho w(e\hs +\hs x) + \sum_{\zeta \in [k]} \nu_\zeta \hs  \left( \hs \binom{n}{\zeta} w(e\hs +\hs x) -\hs  \hs \sum_{\substack{0\leq\alpha\leq e\\ 0\leq\beta\leq o\\ \alpha+\beta \leq \zeta }} \hs {e\choose \alpha} {o \choose \beta}\binom{n\hs - \hs e\hs - \hs o}{\zeta \hs - \hs \alpha \hs - \hs \beta} w (e\hs +\hs x\hs +\hs \beta\hs -\hs \alpha)  \hs \right) \nonumber\\
& && \hspace{4.2in}\forall (e,x,o) \in \mathcal{I}\tag{P$[k]$}\label{opt:dual}
\end{align}
\vspace{-8mm}
\end{figure*}

The smoothness parameters found via Proposition~\ref{prop:smoothLP} can be used with \cref{prop:smooth} to generate lower bounds on the \kstrong price of anarchy.
However, these bounds need not be tight, i.e., there may be no system in the class $\gee_n\times\{w\}$ that attains this inefficiency, and better bounds may be possible.
To study what efficiency we can guarantee across a class of resource allocation problems, we define the \kstrong price of anarchy bound for $(\gee_n,w)$ as
\begin{equation}\label{eq:spoa_bound}
    \spoa_k(\gee_n,w) = \min_{G \in \gee_n} \spoa_k(G,w).
\end{equation}
This performance ratio is parameterized by our choice of welfare function $w$ and the size of collaborative coalitions $k$.
In \cref{thm:spoa}, we provide a linear program whose value provides an exact value of $\spoa_k(\gee_n,w)$.
We do this by showing that the constraints of the $k$ linear programs in Proposition~\ref{prop:smoothLP} can be combined to give an exact quantification of the \kstrong price of anarchy bound.

\begin{theorem}\label{thm:spoa}
For the class of resource allocation problems $\gee_n$ with welfare function $w$, when groups maximize the common interest welfare, then
\begin{equation}
\spoa_k(\gee_n,w) = 1/P^\star(n,w,k),
\end{equation}
where $P^\star(n,w,k)$ is the solution to \eqref{opt:dual}.
\end{theorem}
The proof appears in the appendix.

In \cref{fig:spoa_plot_4}, we consider four welfare functions and plot the tight bounds on the \kstrong price of anarchy for $1 \leq k \leq n$.
As expected, we observe that increased communication improves efficiency guarantees; the amount of this increase is useful in determining the benefits of inter-agent communication/collaboration.
However, this collaboration comes at a cost; in \cref{sec:dyn}, we will study the complexity of distributed dynamics reaching \kstrong Nash equilibria.

\section{Coalitional Dynamics}\label{sec:dyn}
\cref{sec:quant} provided several tools for quantifying the efficiency guarantees of \kstrong Nash equilibria.
In this section, we will study the qualities of group-based dynamics that reach these equilibria.
In particular, we will discuss the convergence rate and transient performance when agents follow the Coalitional round-robin and Asynchronous Best Response, respectively.
We will denote $a^t$ as the joint action occurring at time $t \in \mathbb N$ and $\Gamma^t\subseteq N$ as the group of agents updating their action at time $t$.

\subsection{Round Robin}\label{subsec:rr}
We first consider the $k$-coalitional round robin agent dynamics, in which each group of $k$ agents updates their actions sequentially, following a set order $\sigma \in \Sigma_{\binom{n}{k}}$, where $\sigma(z)$ for $z \in \{1,\ldots,\binom{n}{k}\}$ is the index of a group $\Gamma \in \mathcal{C}_{[k]}$.
We will call a \emph{round} one pass through $\sigma$ in which each group updates their action.
At their turn, the group $\Gamma^t$ selects their best response to the current action, i.e., $a_{\Gamma^t}^{t+1} \in \argmax_{a_{\Gamma^t} \in \mathcal{A}_{\Gamma^t}}W(a_{\Gamma^t},a_{-\Gamma^t})$, where ties are broken uniformly at random unless $a_{\Gamma^t}^{t} \in \argmax_{a_{\Gamma^t} \in \mathcal{A}_{\Gamma^t}}W(a_{\Gamma^t},a_{-\Gamma^t})$, in which case the group selects their current action $a_{\Gamma^t}^{t+1} = a_{\Gamma^t}^{t}$.
The dynamics are more formally described in \cref{alg:krr}.

\begin{algorithm}
\caption{$k$-Round-Robin Dynamics}\label{alg:krr}
\begin{algorithmic}
\Procedure{$k$RoundRobin}{$W,\mathcal{A},N,\sigma,a$}
	\State $\overline{a} \gets$ \texttt{NULL}
	\While{$\overline{a} \neq a$}
		\State $\overline{a} \gets a$
		\For{$z \in \{1,\ldots,\binom{n}{k}\}$}
			\State $\Gamma \gets \mathcal{C}_{k}(\sigma(z))$ \Comment{Get group}
			\For{$a_\Gamma^+ \in \mathcal{A}_\Gamma\setminus a_\Gamma$} \Comment{Group deviations}
				\If{$W(a_\Gamma^+,a_{-\Gamma}) > W(a)$}
					\State $a \gets (a_\Gamma^+,a_{-\Gamma})$
				\EndIf
			\EndFor
		\EndFor
	\EndWhile
\EndProcedure
\end{algorithmic}
\end{algorithm}

These dynamics are synchronous (in that agents must follow a set order) but provide an understanding of how groups of agents can make decisions in a localized manner, and we can analyze the equilibrium hitting time.
In the fully distributed setting ($k=1$), it has been shown that these dynamics reach a Nash equilibrium in finite time and require $\mathcal{O}(m^n)$ welfare evaluations~\cite{durand2016ComplexityOptimalityBest}.
In Proposition~\ref{prop:round_robin}, we find that in the coalitional settings, we maintain the finite convergence time and incur a small base exponential gain in the number of welfare comparisons required.
Recent work has shown that the examples that realize these worst-case hitting times are fragile and that equilibria can be computed in polynomial-time under smoothed running-time analysis~\cite{giannakopoulosSmoothedFPTASEquilibria2023}.
As a first step, we consider the worst-case run time, but the authors believe that similar findings on the added complexity of group decision-making will hold under smoothed running-time analysis, though this is the subject of ongoing work.
Recall $n=|N|$ and $m := \max_{i \in N} |\mathcal{A}_i|$.

\begin{proposition}\label{prop:round_robin}
The $k$-Coalitional-Round-Robin dynamics converge in finite time and requires $\mathcal{O}\left(m^n\left(\frac{1}{1-\sfrac{1}{m}}\right)^k\right)$ welfare evaluations.
% $\mathcal{O}\left(m^n\left(\frac{m}{m-1}\right)^k\right)$
%$\mathcal{O}\left(m^n\frac{m^k}{(m-1)^k}\right)$
\end{proposition}
\begin{proof}
First, we verify that the output of \cref{alg:krr} is a \kstrong Nash equilibrium, then we consider how long it takes \cref{alg:krr} to converge.
\cref{alg:krr} terminates after a round in which no group $\Gamma \in \mathcal{C}_{k}$ can select a new action in which the welfare increases, i.e., $W(a) \geq W(a_\Gamma,a_{-\Gamma})$ for all $a_\Gamma \in \mathcal{A}_\Gamma$ and $\Gamma \in \mathcal{C}_{k}$ where $a$ is the output of \cref{alg:krr}.
A deviation for a any subgroup $\Gamma^\prime \in \mathcal{C}_{[k]}$ is subsumed by the joint action $(a_{\Gamma^\prime},a_{\Gamma\setminus \Gamma^\prime}) \in \mathcal{A}_\Gamma$.
As such, a state $a$ terminates \cref{alg:krr} if and only if it satisfies \eqref{eq:ksnash_def} and is a \kstrong Nash equilibrium.

Without loss of generality, we assume each agent possesses $m$ actions; for each agent, $i$ that has fewer actions, assign $m-|\mathcal{A}_i|$ dummy actions with minimum welfare.
In one round of the $k$-Round-Robin dynamics, each group of agents is given the opportunity to deviate their action.
First, we note that no group $\Gamma$ will respond to the same complimentary group action $a_{-\Gamma}$ in two consecutive rounds unless $a$ is a \kstrong Nash equilibrium.
If the group $\Gamma$ rejects a group action $a_{\Gamma}$ in response to $a_{-\Gamma}$, the joint action $(a_\Gamma,a_{-\Gamma})$ is eliminated from consideration as an output of \cref{alg:krr}.
Accounting for overlaps between the groups, in any round that does not start in a \kstrong Nash equilibrium, at least
$y = \sum_{\zeta=1}^k \binom{n}{\zeta} (m-1)^\zeta,$
joint actions are eliminated as possible outputs of \cref{alg:krr}.
As there are $m^n$ joint actions in total, there can be at most $r \leq \lfloor \frac{m^n}{y} \rfloor +1$ rounds that do not start in a \kstrong Nash equilibrium; this proves the finite convergence time.
In each round, there are exactly $\binom{n}{k}m^k$ welfare checks; thus, the total number of welfare checks is no more than $(\frac{m^n}{y}+1)\binom{n}{k}m^k$.
Removing lower order terms from $y$ gives the stated bound.
\end{proof}

From Proposition~\ref{prop:round_robin}, we observe two things: 1) the coalitional dynamics do not require drastically more welfare evaluations than the fully distributed round-robin, but 2) the convergence rate is slow regardless of $k$.
In light of this, we turn our focus to understanding the transient performance of collaborative decision-making dynamics.
Further, in many settings, it is desirable to allow agents or groups to update their actions asynchronously.
In \cref{subsec:abr}, we will consider both of these factors in the asynchronous best response dynamics.

\subsection{Asynchronous Best-Response Dynamics}\label{subsec:abr}
Motivated by settings where agents (or groups of agents) perform action revisions asynchronously or on their own time scales, we consider a dynamical system where the next group of agents to update is random.

We define the Asynchronous $k$-Coalitional Best-Response Dynamics as follows: let $t \geq 0$ denote the number of agent (or group) updates that have yet occurred\footnote{Counting time steps in terms of the number of updates subsumes cases where agents (or groups) update with respect to individual and independent random clocks. The rate of each clock is analogous to the selection probability for different groups.}.
The updating group $\Gamma^t$ is selected at random, such that the size of the group $\zeta$ is picked with probability $p_\zeta = \mathbb{E}[|\Gamma^t|=\zeta]$ and the specific agents in the group are drawn uniformly at random.
Once formed, the updating group, $\Gamma^t$, chooses their best response in the same manner as the coalitional round robin described in \cref{subsec:rr}.

From their distributed decision-making and asynchronicity, these dynamics capture the behavior of real-time multi-agent systems components. 
In \cref{thm:abr}, we show these dynamics converge almost surely to a \kstrong Nash equilibrium, and further, if the system is \lmksmooth, we provide a bound on the cumulative welfare relative to the optimal.

\begin{theorem}\label{thm:abr}
The Asynchronous $k$-Coalitional Best-Response Dynamics converge almost surely to the set of \kstrong Nash equilibrium.
Further, if $(G,W)$ is a \lmksmooth~system, then after $T \geq 1$ update steps, the cumulative expected welfare satisfies
\begin{equation}\label{eq:br_bound}
%\mathbb{E}\left[ \frac{1}{T} \sum_{t=1}^T  W(a^t)\right] \geq \frac{T-1}{2T} \frac{p^\top \lambda}{1+ p^\top \mu} W(\Opt{a})
\mathbb{E}\left[ \frac{1}{T} \sum_{t=1}^T  W(a^t)\right] \geq \frac{T-1}{2T} \frac{\sum_{\zeta=1}^k p_\zeta \lambda_\zeta}{1+ \sum_{\zeta=1}^k p_\zeta \mu_\zeta} W(\Opt{a}),
\end{equation}
where $p_\zeta$ is the probability a group of size $\zeta$ best responds.
\end{theorem}
Interestingly, the bound on the average transient welfare depends on how frequently groups of different sizes are sampled to perform their best response.
When the agents are designed to more regularly collaborate in larger groups, the transient guarantee will often be better.
\begin{proof}
First, we show that the Asynchronous $k$-Coalitional Best-Response Dynamics converges in general.
A group $\Gamma$ revises their action only to one of strictly higher payoff if one exists.
Consider the resulting Markov chain $\mathcal{M}$ with states $\mathcal{A}$.
Any state $a \in \mathcal{A} \setminus \KSNE$ has an outgoing edge with positive probability as there exists some group $\Gamma \in \mathcal{C}_{[k]}$ that is selected with probability $p_{|\Gamma|}/|\mathcal{C}_{|\Gamma|}| > 0$ which would revise their action.
Any state $a \in \KSNE$ has no outgoing edges with positive probability as no group $\Gamma\in \mathcal{C}_{[k]}$ can revise their action to strictly increase the welfare.
Finally, there are no cycles (excluding self-loops) in $\mathcal{M}$, as every outgoing edge is directed from a joint action of lower welfare to one of strictly higher welfare.
As such, the set $\KSNE$ is absorbing and $\mathbb{P}[\lim_{t \rightarrow \infty} a^t \in \KSNE] = 1$.

Now, consider that the system $(G,W)$ is \lmksmooth.
As the selection of the updating group is random, the welfare at time $t+1$ is a random variable, even when conditioned on $a^t$; the expectation of the succeeding welfare can be written
\begin{align*}
\mathbb{E}[W(a^{t+1})\mid & a^t=a] = \sum_{\zeta=1}^k p_\zeta \sum_{\Gamma \in \mathcal{C}_\zeta} \frac{1}{\binom{n}{\zeta}} W(a_\Gamma^+,a_{-\Gamma})\\
&\geq \sum_{\zeta=1}^k p_\zeta \sum_{\Gamma \in \mathcal{C}_\zeta} \frac{1}{\binom{n}{\zeta}} W(\Opt{a}_\Gamma,a_{-\Gamma})\\
&\geq \sum_{\zeta=1}^k p_\zeta\left( \lambda_\zeta W(\Opt{a}) - \mu_\zeta W(a) \right)\\
&= \left(\sum_{\zeta=1}^k p_\zeta  \lambda_\zeta\right)W(\Opt{a}) - \left(\sum_{\zeta=1}^k p_\zeta  \mu_\zeta\right)W(a),
\end{align*}
where $a_\Gamma^+\in \argmax_{a_{\Gamma} \in \mathcal{A}_{\Gamma}}W(a_\Gamma,a_{-\Gamma})$ is the update state for the group $\Gamma$ following the dynamics; the welfare for each possible updated joint action is the same, so determining which group action is selected is irrelevant. 
As $a_\Gamma^+$ is a best response, the welfare is no better for selecting a different action, namely $\Opt{a}_\Gamma$.
The final inequality holds from \eqref{eq:smooth_def}.
Taking the expectation of $\mathbb{E}[W(a^{t+1})\mid a^t=a]$ over $a^t$ gives
$$\mathbb{E}\left[ W(a^{t+1}) \right] \hs \geq \hs \left(\sum_{\zeta=1}^k p_\zeta  \lambda_\zeta\hspace{-2.5pt}\right) \hs W(\Opt{a}) - \left(\sum_{\zeta=1}^k p_\zeta  \mu_\zeta\hspace{-2.5pt}\right) \hs \mathbb{E}\left[W(a^t)\right]\hspace{-2.5pt}.$$
Rearranging terms shows
\begin{multline*}
\mathbb{E}\left[ W(a^{t+1}) \right] - \frac{\sum_{\zeta=1}^k p_\zeta \lambda_\zeta}{1+ \sum_{\zeta=1}^k p_\zeta \mu_\zeta} W(\Opt{a}) \\
\geq \left(\sum_{\zeta=1}^k p_\zeta \mu_\zeta\right)\hspace{-3pt}\left( \frac{\sum_{\zeta=1}^k p_\zeta \lambda_\zeta}{1+ \sum_{\zeta=1}^k p_\zeta \mu_\zeta} W(\Opt{a}) - \mathbb{E}\left[W(a^t)\right] \right)\hspace{-3.5pt}.
\end{multline*}
Observe that either $\mathbb{E}\left[W(a^t)\right] \geq  \frac{\sum_{\zeta=1}^k p_\zeta \lambda_\zeta}{1+ \sum_{\zeta=1}^k p_\zeta \mu_\zeta} W(\Opt{a})$ or $\mathbb{E}\left[ W(a^{t+1}) \right] - \frac{\sum_{\zeta=1}^k p_\zeta \lambda_\zeta}{1+ \sum_{\zeta=1}^k p_\zeta \mu_\zeta} W(\Opt{a}) \geq 0$.
Accordingly, in expectation, every other update must satisfy the bound, giving the average cumulative welfare bound in \eqref{eq:br_bound}.
\end{proof}

\cref{thm:abr} shows that the transient efficiency changes with the frequency with which different group sizes perform best responses.
To attain the best transient guarantee, we can select $p$ carefully.
\begin{corollary}\label{cor:abr_opt}
If a system $(G,w)$ is a resource allocation problem in $\gee_n\times\{w\}$, then selecting $p_\zeta \propto \frac{\nu_\zeta^\star}{\sum_{\psi\in [k]} \binom{n}{\psi} \nu_\psi^\star}$ for all $\zeta \in [k]$ gives 
$$\mathbb{E}\left[ \frac{1}{T} \sum_{t=1}^T  W(a^t)\right] \geq \frac{T-1}{2T} \spoa_k(\gee_n,w) W(\Opt{a}).$$
\end{corollary}
The proof is omitted as it is straightforward by rearranging terms in the constraints of \eqref{opt:primal}.

Together, \cref{thm:abr} and \cref{cor:abr_opt} provide insight into the transient performance of non-deterministic multi-agent dynamical systems with collaborative communication.
Future work will study the traits of non-best-response dynamics, namely regret-based decision-making.

\subsection{Numerical Example}\label{subsec:sim}
We support the findings of \cref{subsec:abr} by numerical example.
We randomly generate resource allocation problems and simulate the coalitional asynchronous best response dynamics when groups of size $k \in \{1,2,3,4,5\}$ update.

The resource allocation problems are generated by creating 100 resources with values independently drawn uniformly at random on $[0,1]$.
Each of the 25 agents is endowed with between 1 and 10 actions (also sampled uniformly at random).
For each action of each player, each resource is included in that particular action with probability 0.25.
This defines a tuple $G$.
We use the local welfare function $w(x) = xe^{-x/5}$ to capture some added benefit from having multiple agents use the same resource and eventual diminishing returns and increased cost from over congestion.

We select a random initial condition and run the asynchronous best response dynamics with $p_k = 1$ for one value $k \in \{1,2,3,4,5\}$ (i.e., only groups of exactly size $k$ are sampled, but the simulation is repeated for $1 \leq k \leq 5$.
We ran this simulation 100 times.

In \cref{fig:num_sim:group_rev}, we plot the average welfare across the simulations over the number of group action revisions.
We observe that the larger coalitions provide superior transient and long-run performance.
However, a single group action revision requires more computation for larger coalitions.
In \cref{fig:num_sim:welf_check}, for each coalition size $k \in \{1,\ldots,5\}$, we show a scatter plot of the number of cumulative welfare evaluations and the attained system welfare, along with a trend line fit to the data within two standard deviations of the average number of welfare evaluations.
Here, we observe that for lower values of welfare, the smaller coalitions can attain similar welfare with fewer welfare evaluations but that the larger coalitions reach higher welfare much more regularly.

These conclusions help to identify the trade-off in designing systems with collaborative communication: better performance is attainable at the cost of greater computation.

%\begin{figure}
%    \centering
%    \vspace{1mm}
%    \includegraphics[width=0.4825\textwidth]{eg_kBR_MonteCarlo25.PNG}
%    \caption{Numerical simulation {\color{red} have a scatter type plot that has number of welfare checks as horizontal axis, this can be a secondary plot that shows off the two points: requires more welfare checks, but less overall decisions to do well.}}
%    \label{fig:num_sim}
%\end{figure}

\begin{figure}[t!]
\vspace{2mm}
    \centering
    \begin{subfigure}[t!]{0.235\textwidth}
        \includegraphics[width=\textwidth]{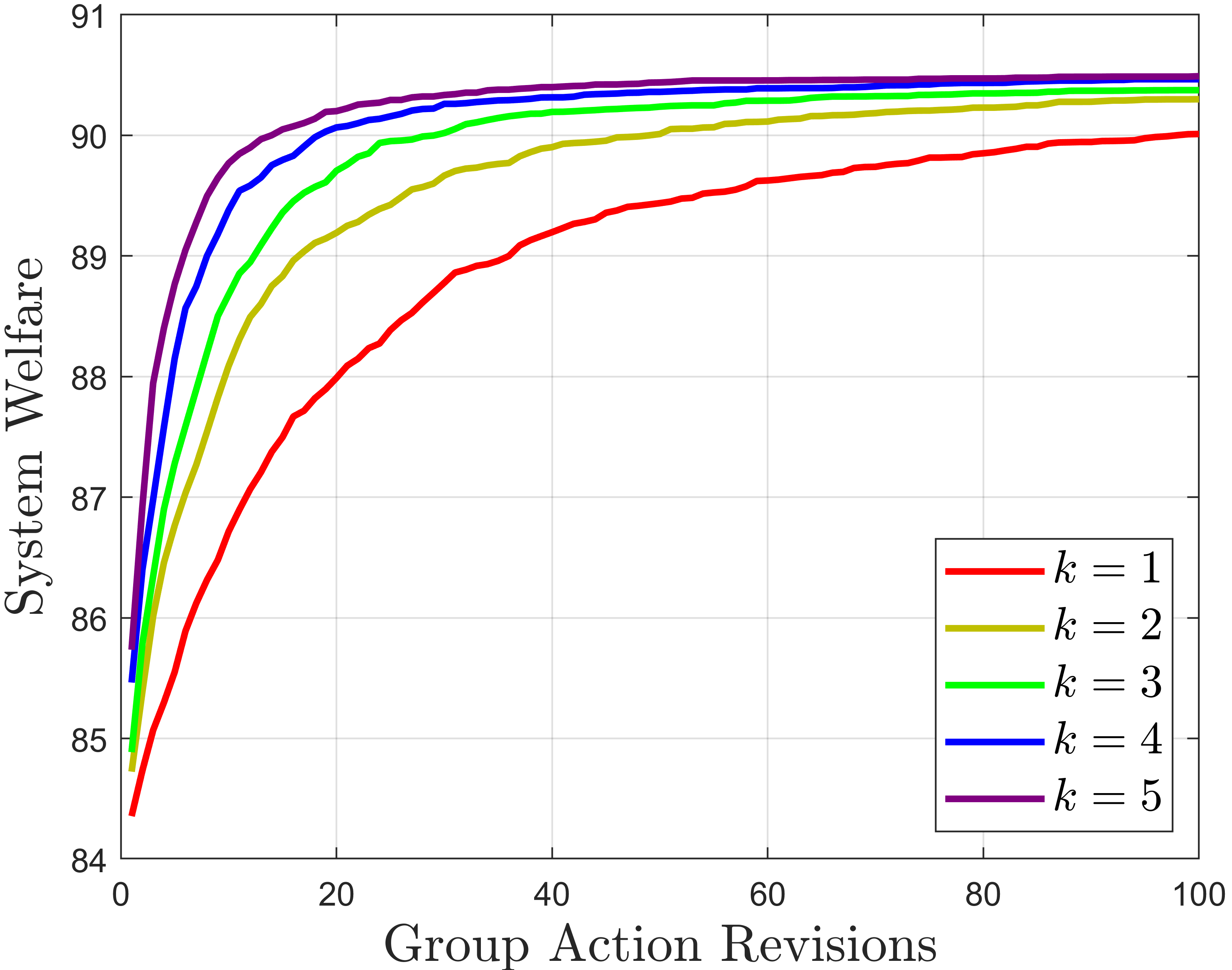}
        \caption{\raggedright \small Group revisions}
        \label{fig:num_sim:group_rev}
    \end{subfigure}
    \begin{subfigure}[t!]{0.235\textwidth}
        \includegraphics[width=\textwidth]{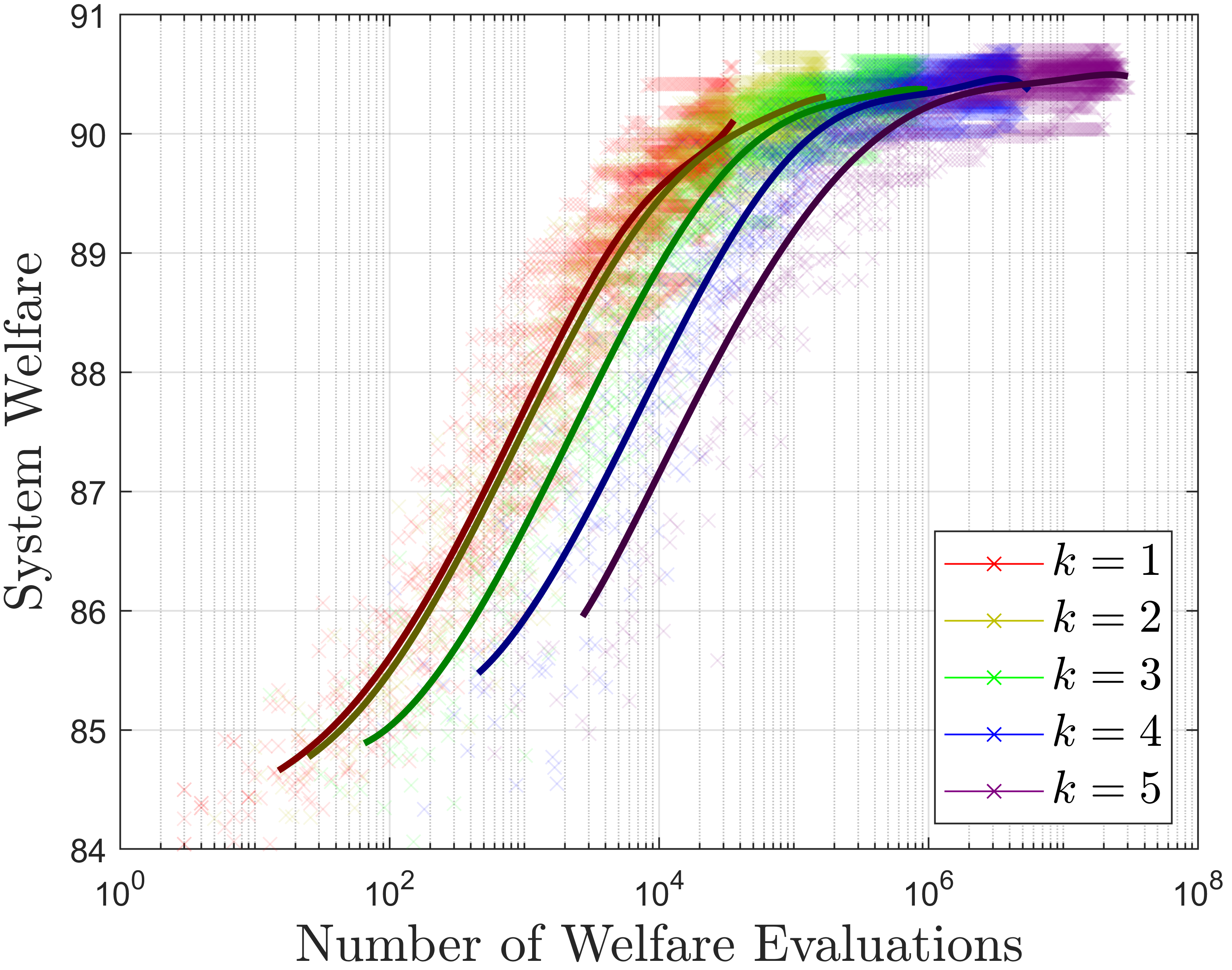}
        \caption{\raggedright \small Welfare checks}
        \label{fig:num_sim:welf_check}
    \end{subfigure}
    \caption{Numerical example of the coalitional asynchronous best response dynamics. In \cref{fig:num_sim:group_rev}, the system welfare is plotted over the number of group action revisions, and in \cref{fig:num_sim:welf_check} it is plotted over the number of welfare evaluations. From this data, we can observe that group revisions offer superior system transient and long-term performance but require more welfare evaluations to compute group actions.}\label{fig:num_sim}
\end{figure}

% Possible figures: Monte Carlo for each dynamic? Relate a sim/figure to resource allocation so it is reoccurring throughout?

\begin{figure*}[t]
    \centering
    \includegraphics[width=0.995\textwidth]{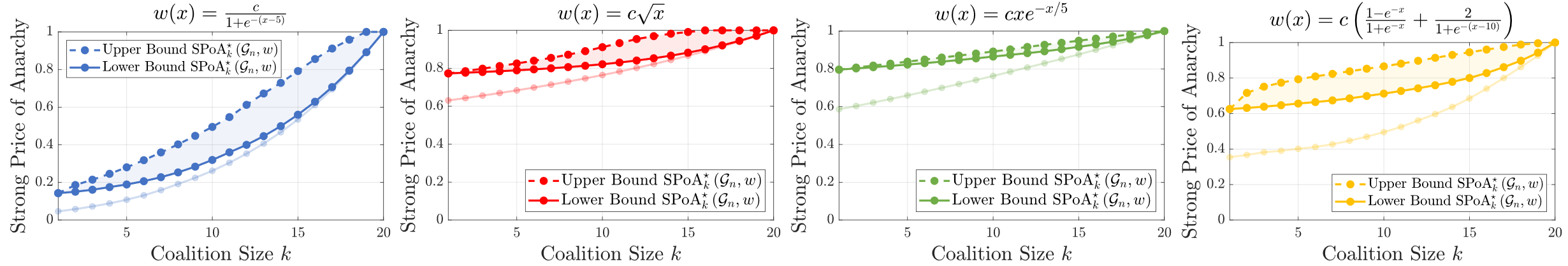}
    \caption{Bounds on the \kstrong Price of Anarchy using the optimal utility function in the class of resource allocation games with welfare function $w$. Upper bound on $\spoa_k^\star(\gee_n,w)$ generated by \cref{prop:util_upper_bound} and lower bound and utility rule that attains it generated by \cref{thm:gsmoothLP}. Compared with the \kstrong price of anarchy when agents optimize the system welfare (lighter line), we demonstrate the possible and guaranteed gain in equilibrium performance attainable by designing group decision-making for collaborative multi-agent systems.}
    \label{fig:spoa_opt_plot_4}
\end{figure*}

\section{Utility Design}\label{sec:design}
Up until this point, agents and groups of agents have been set to optimize the system welfare $W$ over their respective individual or group actions.
Though this is a reasonable approach, the system designer may seek to further improve system performance by designing how a group of agents makes a decision.
Consider that groups of agents instead maximize the objective function $U:\mathcal{A} \rightarrow \mathbb{R}_{\geq 0}$ (henceforth referred to as the \emph{utility function}), i.e.,
\begin{equation}\label{eq:gamma_util_update}
    a_{\Gamma} \in \argmax_{a_\Gamma^\prime \in \mathcal{A}_\Gamma} U(a_\Gamma^\prime,a_{-\Gamma}),
\end{equation}
where ties are still broken at random unless the current group action is in the argmax.
By designing the utility function $U$, the system operator can alter how groups of agents make decisions and, ideally, improve the performance of the system.
A multi-agent system is now captured by the tuple $(G,W,U)$, where the previous results are the special case when $U=W$.

By redefining the objective functions groups of agents seek to maximize, we additionally alter the equilibria that emerge from collaborative decision-making.
We alter the definition of \kstrong Nash equilibria to hold with respect to the utility function, i.e.,
\begin{equation}\label{eq:ksnash_util_def}
    U(\KSNash{a}) \geq U(a^\prime_\Gamma,\KSNash{a}_{-\Gamma}),~\forall a^\prime_\Gamma \in \mathcal{A}_\Gamma,~\Gamma \in \mathcal{C}_{[k]}.
\end{equation}
Let $\KSNE(G,U)$ denote the set of \kstrong Nash equilibria when agents optimize the objective $U$.
The new set of equilibria implies the equilibrium performance guarantee may also change.
As such, we redefine the \kstrong price of anarchy as the approximation of the optimal welfare provided the system equilibria under objective function $U$,
\begin{equation}
    \spoa_k(G,W,U) = \frac{\min_{\KSNash{a} \in \KSNE(G,U)} W(\KSNash{a})}{\max_{\Opt{a} \in \mathcal{A}} W(\Opt{a})}.
\end{equation}

With this new design opportunity, we identify two goals in understanding the new attainable performance of collaborative decision-making: 1) quantifying the performance of a prescribed utility function, and 2) finding a utility function that provides the greatest \kstrong price of anarchy guarantees.
We address these two points in general in \cref{subsec:gsmooth} and more thoroughly within resource allocation problems in \cref{subsec:gra}.

\subsection{Generalized Coalitionally Smooth Games}\label{subsec:gsmooth}
In this section, we consider the general setting and particularly focus on quantifying the \kstrong price of anarchy of a system $(G,W,U)$.
As in \cref{subsec:smooth}, we introduce a notion of smooth systems now generalized to the setting where the agent objective $U$ differs from the system objective $W$.
\begin{definition}\label{def:gsmooth}
A system $(G,W,U)$ is \lmkgsmooth, where $\lambda,\mu \in \mathbb{R}_{\geq 0}^k$, if for all $a,a^\prime \in \mathcal{A}$
\begin{multline}\label{eq:gsmooth_def}
\frac{1}{\binom{n}{\zeta}} \sum_{\Gamma \in \mathcal{C}_\zeta} U(a_\Gamma^\prime,a_{-\Gamma}) - U(a) + W(a)\\
 \geq \lambda_\zeta W(a^\prime) - \mu_{\zeta}W(a),~\forall \zeta \in [k].
\end{multline}
\end{definition}
Like \eqref{eq:smooth_def}, \eqref{eq:gsmooth_def} provides a bound on average deviation effect of a group of size $\zeta$ but on the utility function instead of the welfare.
In Proposition~\ref{prop:gsmooth}, we show that \lmkgsmooth~system permits a bound on the \kstrong price of anarchy.
\begin{proposition}\label{prop:gsmooth}
A system $(G,W,U)$ that is \lmkgsmooth~has $k$-strong price of anarchy satisfying
\begin{equation}\label{eq:gsmooth_prop}
\spoa_k(G,W,U) \geq \frac{\lambda_\zeta}{1+\mu_\zeta},~\forall \zeta \in [k].
\end{equation}
\end{proposition}
\begin{proof}
Let $\KSNash{a} \in \mathcal{A}$ denote a \kstrong Nash equilibrium when agents follow objective function $U$, and let $\Opt{a} \in \argmax_{a \in \mathcal{A}} W(a)$ denote an optimal joint action.
For any $\zeta \in [k]$, we have
\begin{subequations}
\begin{align}
W(\KSNash{a}) &\hs\geq\hs \frac{1}{\binom{n}{\zeta}}\hs \sum_{\Gamma \in \mathcal{C}_\zeta}\hs U(\Opt{a}_\Gamma,\KSNash{a}_{-\Gamma}) \hs -\hs U(\KSNash{a}) \hs +\hs W(\KSNash{a})\label{eq:gsmooth_proof_a}\\
&\geq \lambda_\zeta W(\Opt{a}) - \mu_{\zeta}W(\KSNash{a}).\label{eq:gsmooth_proof_b}
\end{align}
\end{subequations}
Where \eqref{eq:gsmooth_proof_a} holds from $\frac{1}{\binom{n}{\zeta}} \sum_{\Gamma \in \mathcal{C}_\zeta} U(\Opt{a}_\Gamma,\KSNash{a}_{-\Gamma}) - U(\KSNash{a}) \geq 0$ by $\KSNash{a}$ being a \kstrong Nash equilibrium and \eqref{eq:smooth_proof_c} holds from \cref{def:gsmooth}.
Rearranging, we get $W(\KSNash{a})/W(\Opt{a}) \geq \lambda_\zeta/(1+\mu_\zeta)$.
\end{proof}

Beyond quantifying the \kstrong price of anarchy for a system $(G,W,U)$, one may wish to find the utility function which provides the best efficiency guarantee, i.e., $$U \in \argmax_{U^\prime:\mathcal{A}\rightarrow\mathbb{R}_{\geq 0}} \spoa_k(G,W,U^\prime).$$
For a specific problem $(G,W)$, it is possible to design a utility function which guarantees that a system optimal $\Opt{a}$ is a unique equilibrium and provides $\spoa_k(G,W,U)=1$ (e.g., $U(a) = \sum_{i \in N} \indicator[a_i = \Opt{a}_i]$).
However, this would require knowing the optimal allocations a priori, which poses several problems, including: 1) computing an optimal allocation can be intractable, and 2) system parameters may be subject to modeling errors, noise, or changes over time, causing the optimal allocations to change.
As such, we will consider the design of \emph{utility rules}, which provide a set of instructions to construct a utility function across a class of systems and eliminate the computational burden of solving for a new utility function for each system while maintaining improved performance guarantees.
Luckily, the approach in Proposition~\ref{prop:gsmooth} is amenable to generating performance guarantees across a class of systems, and in \cref{subsec:gra}, we will investigate optimal utility rules more thoroughly in resource allocation problems.

%So now, given a $U$, we have a bound on system performance.
%What is the best utility function we could pick?
%For a specific system $(G,W)$, it is possible to design one that guarantees some system optimal $\Opt{a}$ is a unique equilibrium $\spoa_k(G,W,U)=1$ (e.g., $U(a) = \sum_{i \in N} \indicator[a_i = \Opt{a}_i]$) but this would require knowing an optimal joint action a priori.
%This poses two problems 1) computing an optimal allocation can be intractable, and 2) system parameters may be subject to modeling errors, noise, or changes over time, causing the optimal allocations to change.
%The approach in Proposition~\ref{prop:gsmooth} is amenable to generating performance guarantees across a class of systems (all those that satisfy \eqref{eq:gsmooth_prop}).
%We consider this more thoroughly in resource allocation games, where we consider the design of a \emph{utility rule} that can be easily implemented in any such system and eliminates the computational burden of solving for a new utility function for each problem while maintaining good performance guarantees.

\subsection{Resource Allocation Games}\label{subsec:gra}
In this section, we consider the \kstrong price of anarchy in classes of resource allocation problems when the agents' objective is derived from a utility rule $u \in \mathbb{R}^{n+1}_{\geq 0}$.
In an agent environment $G = (N,\mathcal{A},\mathcal{R},\{v_r\}_{r \in \mathcal{R}})$, the utility rule $u$ can be applied to derive the utility function
$$U(a) = \sum_{r \in \mathcal{R}} v_r u(|a|_r).$$
To normalize the utility function, we set $u(0)=0$.
We ultimately consider the performance of a utility rule $u$ across all agent environments $G \in \gee_n$ with welfare function $w$.
We slightly abuse notation to refer to a system by the tuple $(G,w,u)$.
To quantify this performance, we generalize the \kstrong price of anarchy bound defined in \cref{subsec:rr} to hold for cases where groups of agents optimize the utility function.
\begin{equation}
\spoa_k(\gee_n,w,u) = \min_{G \in \gee_n} \spoa_k(G,w,u).
\end{equation}
The performance ratio is parameterized by the pair $(w,u)$; as such, we will discuss the effectiveness of a utility rule $u$ with respect to a given welfare function  $w$.

Taking the utility rule approach completely eliminates the computational cost of deriving a utility function for each problem instance; now we seek to understand the capabilities of this approach in two ways: 1) in \cref{thm:gsmoothLP} we demonstrate how we can construct utility rules with good performance guarantees, and 2) in Proposition~\ref{prop:gsmooth} we provide an upper bound on the best attainable performance a utility rule can provide.
In \cref{cor:design}, we provide a formal condition on when the constructed utility rule is optimal.

\begin{theorem}\label{thm:gsmoothLP}
Any resource allocation problem $(G,W) \in \gee_n\times\{w\}$ with the utility rule $\widetilde{u}_\zeta$ is $(1,\widetilde{\rho}_\zeta-1)$-$k$-generalized-coallitionally smooth, where $\widetilde{u}_\zeta$ and $\widetilde{\rho}_\zeta$ are solutions to the linear program,
\begin{align}
&(\widetilde{\rho}_\zeta,\widetilde{u}_\zeta)\in \argmin_{\rho\geq 0, u \in \mathbb{R}^{n+1}_{\geq 0}} ~~~ \rho \nonumber\\
&{\rm s.t.} \hspace{8pt}0 \geq w(o\hs +\hs x)-\rho w(e\hs +\hs x) + \nonumber\\
&~~ \left( \hspace{-1.2pt} \hs \binom{n}{\zeta} u(e\hs +\hs x) \hs -\hs  \hs \sum_{\substack{0\leq\alpha\leq e\\ 0\leq\beta\leq o\\ \alpha+\beta \leq \zeta }} \hs {e\choose \alpha} {o \choose \beta}\binom{n\hs - \hs e\hs - \hs o}{\zeta \hs - \hs \alpha \hs - \hs \beta} u (e\hs +\hs x\hs +\hs \beta\hs -\hs \alpha)  \hs \right)\nonumber\\
&\hspace{160pt}\forall (e,x,o) \in \mathcal{I}.\tag{Q$\zeta$}\label{opt:Gsmooth}
\end{align}
\end{theorem}
\begin{proof}
Consider the parameterization described in the proof of Proposition~\ref{prop:smoothLP}, where for any two actions $a,a^\prime \in \mathcal{A}$, we can rewrite $W(a) = \sum_{e,x,o} \theta(e,x,o)w(e+x)$ and $W(a^\prime) = \sum_{e,x,o}\theta(e,x,o)w(o+x)$.
Now, we can additionally rewrite $U(a) = \sum_{e,x,o} \theta(e,x,o)u(e+x)$ and 
\begin{multline*}
\sum_{\Gamma \in \mathcal{C}_\zeta}  W(a^\prime_\Gamma,a_{-\Gamma})\\ = \sum_{e,x,o} \theta(e,x,o) \sum_{\substack{0 \leq \alpha \leq e\\ 0\leq \beta \leq o \\ \alpha+\beta \leq \zeta}} \hspace{-2mm} {e\choose \alpha} {o \choose \beta}\binom{n\hs - \hs e\hs - \hs o}{\zeta \hs - \hs \alpha \hs - \hs \beta}w(e \hspace{-1mm} + \hspace{-1mm} x \hspace{-1mm} + \hspace{-1mm} \beta \hspace{-1mm} - \hspace{-1mm} \alpha).
\end{multline*}

We can now write out \eqref{eq:gsmooth_def}, the \lmkgsmooth~constraint, as
\begin{multline*}
\sum_{e,x,o} \theta(e,x,o) \Bigg( \frac{1}{\binom{n}{\zeta}}\sum_{\substack{0 \leq \alpha \leq e\\ 0\leq \beta \leq o \\ \alpha+\beta \leq \zeta}} \hspace{-2mm} {e\choose \alpha} \hs {o \choose \beta} \hs \binom{n\hs - \hs e\hs - \hs o}{\zeta \hs - \hs \alpha \hs - \hs \beta}u(e  \hspace{-0.5mm}+ \hspace{-1mm} x \hs + \hs \beta \hs - \hs \alpha)\\ 
- u(e+ x)\hs\Bigg)
 \geq \sum_{e,x,o}\theta(e,x,o)\left(\lambda_\zeta w(o\hs +\hs x) \hs - \hs (\mu_\zeta\hs +\hs 1) w(e\hs +\hs x)\right).
\end{multline*}
As before, we can observe that this constraint is sufficiently satisfied when 
\begin{multline}\label{eq:gsmooth_comp_constraint}
\frac{1}{\binom{n}{\zeta}}\sum_{\substack{0 \leq \alpha \leq e\\ 0\leq \beta \leq o \\ \alpha+\beta \leq \zeta}} \hspace{-2mm} {e\choose \alpha} \hs {o \choose \beta} \hs \binom{n\hs - \hs e\hs - \hs o}{\zeta \hs - \hs \alpha \hs - \hs \beta}u(e  \hs + \hs x \hs + \hs \beta \hs - \hs \alpha) - u(e+x)\\
 \geq \lambda_\zeta w(o\hs +\hs x) \hs - \hs (\mu_\zeta\hs +\hs 1) w(e\hs +\hs x),\quad \forall (e,x,o) \in \mathcal{I}.
\end{multline}

The task of finding smoothness parameters that give the best price of anarchy guarantee becomes the same problem as \eqref{opt:smooth_comp} but now with constraint set \eqref{eq:gsmooth_comp_constraint}.
By substituting the decision variables $\rho = (1+\mu_\zeta)/\lambda_\zeta$ and $\nu = 1/\left(\binom{n}{\zeta}\lambda_\zeta\right) \geq 0$, we attain the new constraint set
\begin{multline}
0 \geq w(o\hs +\hs x)-\rho w(e\hs +\hs x) +\\
\nu\left( \hspace{-1.2pt} \hs \binom{n}{\zeta} u(e\hs +\hs x) \hs -\hs  \hs \sum_{\substack{0\leq\alpha\leq e\\ 0\leq\beta\leq o\\ \alpha+\beta \leq \zeta }} \hs {e\choose \alpha} {o \choose \beta}\binom{n\hs - \hs e\hs - \hs o}{\zeta \hs - \hs \alpha \hs - \hs \beta} u (e\hs +\hs x\hs +\hs \beta\hs -\hs \alpha)  \hs \right)\\
\forall (e,x,o) \in \mathcal{I}.
\end{multline}
The new objective\footnote{As an aside, the transformed program up to this point can be used to evaluate the performance of a specified utility rule.} becomes $1/\rho$.

Finally, we let $u \in \mathbb{R}^n_{\geq 0}$ become a decision variable in the program.
Observe that every occurrence of $u$ is multiplied by $\nu$, and every occurrence of $\nu$ multiplies $u$.
As such, we can define the new decision variable $u^\prime = \nu u$ and retrieve the linear program \eqref{opt:Gsmooth}.
\end{proof}

The utility rule $\hat{u}_\zeta$ that \eqref{opt:Gsmooth} provides us some guarantee on attainable performance from designing group decision-making in collaborative systems.
However, it is not yet clear if these are the best possible utility rules.
To understand what the best possible performance is of a collaborative system, we define the optimal \kstrong price of anarchy as
\begin{equation}
\spoa_k^\star(\gee_n,w) = \sup_{u:[n] \rightarrow{R}_{\geq 0}} \spoa_k(\gee_n,w,u).
\end{equation}
This upper bound informs us of what efficiency is possible to hope for out of a collaborative system.
In Proposition~\ref{prop:util_upper_bound}, we bound this quantity.

\begin{proposition}\label{prop:util_upper_bound}
For the class of resource allocation problems $\gee_n\times\{w\}$, when agents maximize the optimal utility design objective $u^\star$,
\begin{equation}
\spoa_k^\star(\gee_n,w) \leq 1/Q^\star(n,w,k),
\end{equation}
where $Q^\star(n,w,k)$ is value of the linear program
\begin{align}
&Q^\star(n,w,k) = \min_{\rho\geq 0, \{u_\zeta \in \mathbb{R}^{n+1}_{\geq 0}\}_{\zeta \in [k]}} ~~~ \rho \nonumber\\
&{\rm s.t.} \hspace{8pt}0 \geq w(o\hs +\hs x)-\rho w(e\hs +\hs x) + \nonumber\\
&\sum_{\zeta\in [k]}\hs\left( \hspace{-1.2pt} \hs \binom{n}{\zeta} u_\zeta(e\hs +\hs x) \hs -\hs \hs  \hs \sum_{\substack{0\leq\alpha\leq e\\ 0\leq\beta\leq o\\ \alpha+\beta \leq \zeta }}\hs \hs {e\choose \alpha} \hs {o \choose \beta} \hs\binom{n\hs - \hs e\hs - \hs o}{\zeta \hs - \hs \alpha \hs - \hs \beta} u_\zeta (e\hs +\hs x\hs +\hs \beta\hs -\hs \alpha)  \hs \right)\nonumber\\
&\hspace{140pt}\forall (e,x,o) \in \mathcal{I}.\tag{Q$[k]$}\label{opt:util_opt}
\end{align}
\end{proposition}
The proof appears in the appendix.

Note that \cref{thm:gsmoothLP} provides a utility rule with associated performance guarantee which lower bounds $\spoa_k^\star(\gee_n,w)$, and Proposition~\ref{prop:util_upper_bound} provides an upper bound.
In \cref{cor:design}, we note that when these two bounds match, we have a tight bound on $\spoa_k^\star(\gee_n,w)$ as well as an optimal utility rule.

\begin{corollary}\label{cor:design}
For the class of resource allocation problems $\gee_n\times\{w\}$, if the value of \eqref{opt:Gsmooth} satisfies $\rho^\star_\zeta = Q^\star(n,w,k)$, then $\spoa_k^\star(\gee_n,w) = 1/Q^\star(n,w,k)$ is a tight bound and a solution $\widetilde{u}_\zeta$ to \eqref{opt:Gsmooth} is an optimal utility rule.
\end{corollary}
\begin{proof}
This follows immediately from $1/\rho_\zeta^\star = \frac{\lambda_\zeta}{1+\mu_\zeta}$ being a lower bound on $\spoa_k^\star(\gee_n,w)$ and the reciprocal of the value of \eqref{opt:util_opt}, $1/Q^\star$ being an upper bound.
When the two match, the bound must be tight.
\end{proof}
The two bounds coinciding is not guaranteed but does occur at the extremes ($k=1$ and $k=n$); further, the gap between the two bounds (if present) is often small, and the lower bound attained by the utility rule constructed in \cref{thm:gsmoothLP} often demonstrates a significant improvement over the setting where agents simply optimize the system objective.
Consider the four welfare functions from \cref{fig:spoa_plot_4} again; for each, we find that the utility rule computed using \cref{thm:gsmoothLP} and the upper bound on $\spoa_k^\star(\gee_n,w)$ using \cref{prop:util_upper_bound}.
In \cref{fig:spoa_opt_plot_4} we plot these lower and upper bounds on $\spoa_k^\star(\gee_{20},w)$ for each utility function and for each value of $1 \leq k \leq n$; these values are juxtaposed with the \kstrong price of anarchy when agents optimize the system objective $w$ to demonstrate the possible gain in performance from designing the agents' objective in collaborative systems.

\section{Conclusion}
In this work, we provided a variety of tools for evaluating the benefits and costs of collaborative communication in multi-agent systems.
A collaborative multi-agent system was modeled by a common interest game where groups of players collaboratively perform their best responses simultaneously.
We specifically considered the \kstrong Nash equilibrium as a relevant equilibrium concept to gain insights into system behavior between the fully centralized and fully distributed settings.
We introduced the notion of \lmksmooth~systems and derived bounds on how well the \kstrong Nash equilibrium approximates the optimum in such systems.
Further analysis studied the running time of collaborative multi-agent decision dynamics and their transient performance, as well as the possible performance gains from designing agents' objectives separately from the system objective.
Finally, we underwent a more thorough study in the class of resource allocation games, in which we provided tractable linear programs whose solutions give tight bounds on the \kstrong price of anarchy in resource allocation games.
Future work will study less extensive communication paradigms and dynamical systems that emerge when agents learn together.

\section*{References}
\vspace{-12pt}
\bibliographystyle{IEEEtran}
%\bibliography{../../../My_Library}

\appendix
\noindent\emph{Proof of Proposition~\ref{prop:exist}:}
To show existence, we can simply observe that $\Opt{a} \in \argmax_{a \in \mathcal{A}} W(a)$ is a \kstrong Nash equilibrium for any $k \in [n]$.
Because $W(\Opt{a})\geq W(a^\prime)$ for all $a^\prime \in \mathcal {A}$, the global optimal satisfies
$W(\Opt{a}) \geq W(a_{\Gamma}^\prime,\Opt{a}_{-\Gamma}),~\forall a_\Gamma^\prime \in \mathcal{A}_\Gamma,~\Gamma \in \mathcal{C}_{[k]}$. \hfill \qed

\vspace{6pt}
\noindent\textit{Proof of \cref{thm:spoa}}:
The proof can be outlined in four parts: first, the problem of finding $\spoa_k(\gee_n,w)$ is transformed and relaxed; second, the parameterization used in the proof of Proposition~\ref{prop:smoothLP} is used to turn the relaxed problem into a linear program. Next, an example is constructed to show the linear program provides a tight bound. Finally, we take the dual of said linear program.

\begin{figure}
    \centering
    \vspace{1mm}
    \includegraphics[width=0.4825\textwidth]{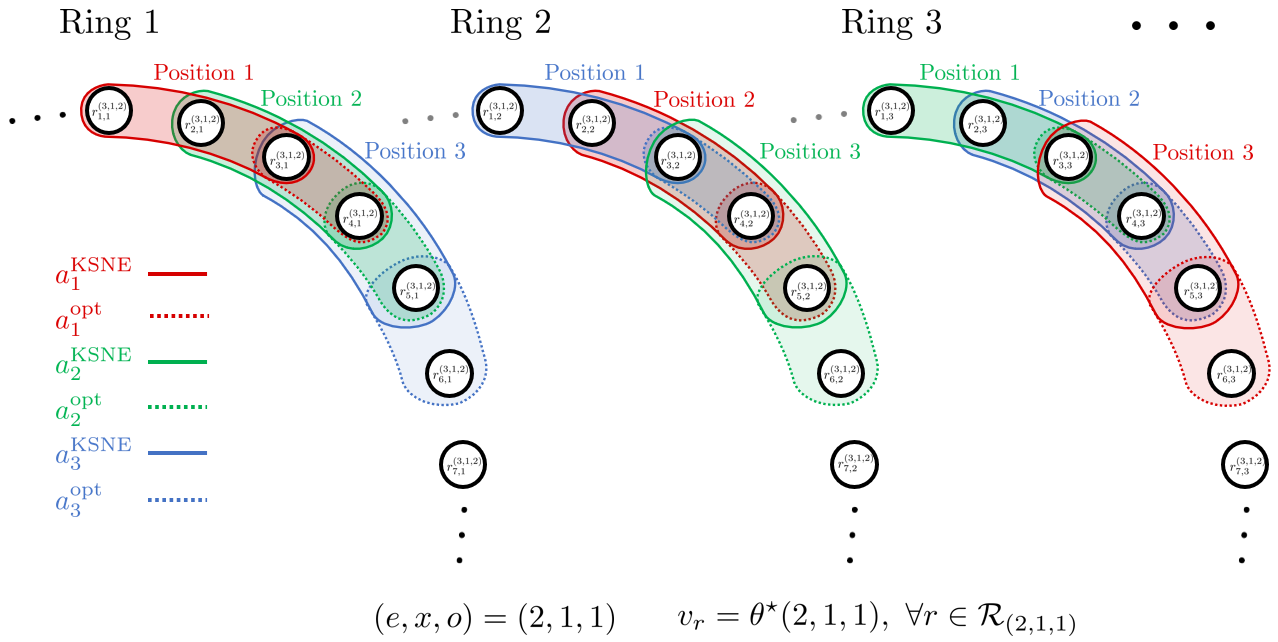}
    \caption{Game construction for worst-case \kstrong price of anarchy. Three of the $n$ players' action sets are shown (color-coded in red, green, and blue, respectively) on three of $n!$ rings for the label $(e,x,o)=(2,1,1)$. A ring has $n$ positions, one for each player. For a label $(e,x,o)$ we generate $n!$ rings for all the orderings of players over positions. This is repeated for each label. Players still only have two actions, but each action covers resources from each ring. The value of a resource is equal to the value of $\theta^\star$, a solution to \eqref{opt:primal}, for the label with which it is associated.}
    \label{fig:rings}
\end{figure}

\begin{figure*}
%\begin{maxi}|s|
%{\scriptstyle \theta \in \mathbb{R}^{|\mathcal{I}|}_{\geq 0}}{\sum_{e,x,o}w(o+x)\theta(e,x,o)}
%{\tag{P}\label{opt:primal}}{}
%\addConstraint{\sum_{e,x,o}\left( \frac{n!}{(n-\zeta)!}w(e+x)- \sum_{\substack{0\leq\alpha\leq \zeta\\ 0\leq\beta\leq \zeta-\alpha}} {\zeta\choose \alpha} {\zeta-\alpha \choose \beta} e^{\underline{\alpha}}o^{\underline{\beta}}(n\hs-\hs e\hs -\hs o)^{\underline{\zeta\text{-}\alpha\text{-}\beta}} w(e\hs +\hs x\hs +\hs \beta\hs -\hs \alpha)  \right)\theta(e,x,o) \geq 0}
%\addConstraint{~~~~~~~~~~~~~~~~~~~~~~~~~~~~~~~~~~~~~~~~~~~~~~~~~~~~~~~~~~~~~~~~~~~~~~~~~~~~~~~~~~~~~~~~~~~~~~~~~\forall\zeta \in \{1,\ldots,k\}}
%\addConstraint{\sum_{e,x,o}w(e+x)\theta(e,x,o)=1}
%\end{maxi}
\begin{align}
\max_{\theta \in \mathbb{R}^{|\mathcal{I}|}_{\geq 0}} \quad & \sum_{e,x,o}w(o+x)\theta(e,x,o)\nonumber\\
~~{\rm s.t.} & \sum_{e,x,o}\left( \binom{n}{\zeta}w(e+x)-  \sum_{\substack{0\leq\alpha\leq e\\ 0\leq\beta\leq o\\ \alpha+\beta \leq \zeta }} \hs {e\choose \alpha} {o \choose \beta}\binom{n\hs - \hs e\hs - \hs o}{\zeta \hs - \hs \alpha \hs - \hs \beta} w (e\hs +\hs x\hs +\hs \beta\hs -\hs \alpha)  \right)\theta(e,x,o) \geq 0\nonumber\\
 & ~~~~~~~~~~~~~~~~~~~~~~~~~~~~~~~~~~~~~~~~~~~~~~~~~~~~~~~~~~~~~~~~~~~~~~~~~~~~~~~~~~~~~~~~~~~~~~~~~\forall\zeta \in \{1,\ldots,k\}\nonumber\\
 & \sum_{e,x,o}w(e+x)\theta(e,x,o)=1\tag{D}\label{opt:primal}
\end{align}
\vspace{-8mm}
\end{figure*}

\noindent\emph{1) Relaxing the problem:} Quantifying $\spoa_k(\gee_n,w)$ can be expressed as taking the minimum \kstrong price of anarchy over all games in $\gee_n$, i.e.,
\begin{mini}|s|
{\scriptstyle G \in \gee_n}{\frac{\min_{\KSNash{a} \in \KSNE(G)} W(\KSNash{a})}{\max_{\Opt{a} \in \mathcal{A}} W(\Opt{a})}}
{\tag{D1}\label{opt:tproof1}}{}
\end{mini}
To make this problem more approachable, we introduce several transformations and relaxations.
First, rather than searching over the entire set of game $\gee_n$, we search over the set of games $\hat{\gee}_n$, in which each agent has exactly two actions.
This reduction of the search space can be done without loss of generality, i.e., $\spoa_k(\gee_n,w) = \spoa_k(\hat{\gee}_n,w)$.
Trivially, $\hat{\gee}_n \subset \gee_n$.
Further, consider any game $G \in \gee_n$; if for every player, each of their actions is removed except their action in the optimal allocation $\Opt{a}_i$ and their action in their worst \kstrong Nash equilibrium $\KSNash{a}_i$, the new problem will maintain the same \kstrong price of anarchy, but will now exist in $\hat{\gee}_n$.
With this reduction, we will denote each player's action set as $\hat{\mathcal{A}}_i = \{\Opt{a}_i,\KSNash{a}_i\}$.
Second, we normalize each resource value $v_r$ such that the equilibrium welfare is one.
This, too, can be done without loss of generality by scaling each resource identically, thus not altering the $\spoa$ ratio.
Third, we invert the objective and consider the maximization of $W(\Opt{a})/W(\KSNash{a})$.
Finally, we sum over each of the $k$-coalition equilibrium constraints.
For each $\zeta \in [k]$, rather than satisfying each inequality in \eqref{eq:ksnash_def}, sum over every combination of the $\zeta$ out of $n$ players, denoted $\mathcal{C}_\zeta$.
Applying these reductions to \eqref{opt:tproof1} gives,
\begin{maxi}|s|
{\scriptstyle G \in \hat{\gee}_n}{W(\Opt{a})}
{\tag{D2}\label{opt:tproof2}}{}
\addConstraint{\hs\hs\binom{n}{\zeta}W(\KSNash{a}) \geq \sum_{\Gamma \in \mathcal{C}_\zeta} W(\Opt{a}_\Gamma,\KSNash{a}_{-\Gamma}),~\forall \zeta \in [k]}
%\addConstraint{~~~~~~~~~~~~~~~~~~~~~~~~~~~~~~~~~~~~~~~~~~~\forall \zeta \in [k]}
\addConstraint{W(\KSNash{a}) = 1}{}
\end{maxi}
\eqref{opt:tproof2} provides a lower bound on $\spoa_k(\gee_n,w)$ as the feasible set was expanded.
Later, we will show that the bound is tight by constructing an example that realizes it.

\noindent\emph{2) Parameterization:} We use the parameterization introduced in the proof of Proposition~\ref{prop:smoothLP} with respect to the joint actions $a = \KSNash{a}$ and $a^\prime = \Opt{a}$.
By considering any $\theta \in \mathbb{R}_{\geq 0}^{|\mathcal{I}|}$, we can parameterize any game $G \in \hat{\mathcal{G}}_n$; to find the worst-case price of anarchy, we search over all such parameters, i.e., look over the entire class of games.
The linear program \eqref{opt:primal} is the result of the search for the vector $\theta$ that results in the highest price of anarchy.

\noindent\emph{3) Constructing an example:} Consider the following resource allocation problem: for each label $(e,x,o) \in \mathcal{I}$ and permutation of the $n$ player $\sigma \in {\Sigma}_n$, define a ring of $n$ resources.
Total, there are $nn!|\mathcal{I}|$ resources.
Let $r_{i,j}^{(e,x,o)}$ denote the resource with label $(e,x,o)$ at position $i$ in the $j$th ring.
Consider, for instance, the $n!$ rings associated with the label $(e,x,o)=(2,1,1)$ as depicted in \cref{fig:rings}.
We will construct the actions $\KSNash{a}_i$ and $\Opt{a}_i$ so that for each resource in these rings, $e+x=3$ agents have it in only their equilibrium action, and $x+o=2$ agents have it only in their optimal action.
In the first ring (with the monotonic permutation $\sigma = (1,2,3,\ldots,n)$), agent $i$ has actions $\KSNash{a}_i = \{r^{(2,1,1)}_{i,1}, r^{(2,1,1)}_{i+1 \% n,1}, r^{(2,1,1)}_{i+2 \% n,1}\}$ and $\Opt{a}_i = \{r^{(2,1,1)}_{i+2,1}, r^{(2,1,1)}_{i+3 \% n,1}\}$, where $\%$ denotes the modulo operator so the selected resources wrap around the ring.
This pattern continues for each ring $j \in [n!]$ with a different permutation of players $\sigma \in {\Sigma}_n$.
At a ring with label $(e,x,o)$ and permutation $\sigma$, player $i$ has the actions $\KSNash{a}_i = \{r^{(e,x,o)}_{\sigma(i),j},\ldots, r^{(e,x,o)}_{\sigma(i)+e+x-1 \% n,j}\}$ and $\Opt{a}_i = \{r^{(e,x,o)}_{\sigma(i)+e \% n,j},\ldots, r^{(e,x,o)}_{\sigma(i)+e+x+o-1 \% n,j}\}$.
Finally, each resource of type $(e,x,o)$ has a value $\theta(e,x,o)$ where $\theta$ is a fixed parameter.
The function which encodes the welfare from player overlap is $w$.

In the joint action $\KSNash{a}$, each resource is covered by exactly $e+x$ agents, and the system welfare can be written
\begin{equation}\label{eq:W_ksne_welf}
W(\KSNash{a}) = \sum_{e,x,o}n n!\theta(e,x,o)w(e+x).
\end{equation}
Similarly, joint action $\Opt{a}$ satisfies
\begin{equation}\label{eq:W_opt_welf}
	W(\Opt{a}) = \sum_{e,x,o} n n!\theta(e,x,o)w(o+x).
\end{equation}
Now, consider a coalition $\Gamma \in \mathcal{C}_{[k]}$ and denote by $\zeta$ its cardinality.
The system welfare of this group deviating their action to $\Opt{a}_\Gamma$ is
\begin{align}
&W(\Opt{a}_\Gamma,\KSNash{a}_{-\Gamma}) = \sum_{e,x,o}\sum_{j=1}^{n!} \sum_{i=1}^n \theta(e,x,o)w(|\Opt{a}_\Gamma,\KSNash{a}_{-\Gamma}|_r)\nonumber\\
&=\sum_{e,x,o}  \theta(e,x,o)\hs \sum_{\substack{0\leq\alpha\leq e\\ 0\leq\beta\leq o\\ \alpha+\beta \leq \zeta }} \hs\hs nn!  {e\choose \alpha} {o \choose \beta}\binom{n\hs - \hs e\hs - \hs o}{\zeta \hs - \hs \alpha \hs - \hs \beta} w (e\hs +\hs x\hs +\hs \beta\hs -\hs \alpha)\label{eq:W_gamma_dev}
\end{align}
where we let $r$ be the shorthand for $r_{i,j}^{(e,x,o)}$.
The second equality holds by defining $\alpha$ and $\beta$ as the number of players in $\Gamma$ who invested in resource $r$ exclusively in their action $\KSNash{a}$ or $\Opt{a}$ respectively.
By counting arguments, there are exactly $\binom{e}{\alpha}\binom{o}{\beta}\binom{n-e-o}{\zeta-\alpha-\beta}$ positions for the players in $\Gamma$ which yield the profile $(\alpha,\beta)$ for a resource at some fixed position in the ring, there are $\zeta!$ ways to order the players in $\Gamma$, $(n-\zeta)!$ ways to order the players not in $\Gamma$, and $n$ resource in each ring.

Verifying $\KSNash{a}$ is a \kstrong Nash equilibrium boils down to showing \eqref{eq:W_ksne_welf} is greater than or equal to \eqref{eq:W_gamma_dev}.
We can see that this holds whenever $\theta$ is a feasible point in \eqref{opt:primal}.
Accordingly, the \kstrong price of anarchy satisfies
\begin{equation}
\frac{1}{Q^\star} \leq \spoa_k(\gee_n,w) \leq \frac{1}{\sum_{e,x,o}\theta(e,x,o)w(o+x)},
\end{equation}
where the first inequality holds from the reductions made in part 1, and the second holds as the \kstrong price of anarchy is upper bounded by any particular problem; comparing \eqref{eq:W_ksne_welf} and \eqref{eq:W_opt_welf} gives the final expression.
Letting $\theta$ take on the solution to \eqref{opt:primal} shows the bound is tight.

\noindent\emph{4) Taking the Dual:} Before considering the dual program to \eqref{opt:primal}, we first show that the primal is feasible.
It is easy to verify the feasible set is non-empty by considering the point $\theta(1,0,0) = 1/w(1)$ and zero otherwise.
Now, we must show that the feasible set is compact, and thus, the value of \eqref{opt:primal} is bounded.
From the equality constraint, we can obtain $$1 \geq \min_{y>0} w(y) \sum_{\substack{e,x,o\\e+x>0}}\theta(e,x,o).$$
Because we assume $w(y)>0$ for all $y>0$, we show that each value of $\theta(e,x,o)$ such that $e+x>0$ is bounded.
For the remaining values of $\theta(0,0,o)$, consider the equilibrium constraint\footnote{The $\zeta=1$ constraint is present in \eqref{opt:primal} for all $k \geq 1$.} when $\zeta=1$.
By rearranging terms and observing the bounded terms from the previous argument, we observe $L \geq w(1) \sum_{o \in [n]} o \theta(0,0,o)$, where $L$ is a bounded value.
Because $w(1)>0$, the remaining decision variables are also bounded, and thus the feasible set is finite.

Now, we find the dual program to \eqref{opt:primal}.
Because \eqref{opt:primal} is a linear program, we can rewrite it in the more concise form
\begin{maxi}|s|
{\scriptstyle \theta \in \mathbb{R}^{|\mathcal{I}|}}{b^\top \theta\nonumber}
{}{}
\addConstraint{c_\zeta^\top \theta }{\geq 0,~\forall \zeta \in [k]}{\quad\quad(\nu_\zeta)}
\addConstraint{d^\top \theta - 1 }{= 0}{\quad\quad(\rho)}
\addConstraint{\theta}{\geq 0}{\quad\quad(\phi)}
\end{maxi}
where $\nu \geq 0$, $\rho$, and $\phi \geq 0$ are the associated dual variables.
The Lagrangian function is defined as $\mathcal{L}(\theta,\nu,\rho,\phi) = b^\top \theta + (\sum_{\zeta \in [k]} \nu_\zeta c_\zeta^\top\theta) - \rho (d^\top \theta -1) + \phi^\top \theta$.
Let $g(\nu,\rho,\phi) = \sup_{\theta \in \mathbb{R}^{|\mathcal{I}|}} \mathcal{L}(\theta,\nu,\rho,\phi)$ serve as an upper bound to \eqref{opt:primal}.
The dual program is derived by minimizing $g(\nu,\rho,\phi)$; note that this value is only unbounded above unless $b^\top + \sum_{\zeta \in [k]}\nu_\zeta c_\zeta^\top - \rho d^\top + \phi^\top = 0$.
Substituting this into the objective and removing the free variable $\phi$ so that the equality constraint becomes an inequality, the dual problem becomes
\begin{mini}|s|
{\scriptstyle \rho,\{\nu_\zeta \in \mathbb{R}_{\geq 0}\}_{\zeta \in [k]}}{\rho}
{\tag{P1}\label{opt:design1}}{}
\addConstraint{b^\top - \rho d^\top + \sum_{\zeta \in [k]} \nu_\zeta c_\zeta }{\leq 0}
\addConstraint{}
\end{mini}
From strong duality, \eqref{opt:design1} provides the same value as \eqref{opt:primal}.
Expanding terms show that \eqref{opt:design1} is equivalent to \eqref{opt:dual}.
\hfill\qed

\vspace{6pt}
\noindent\textit{Proof of Proposition~\ref{prop:util_upper_bound}}:
The proof is straightforward and simply requires generalizing the constraint set of \eqref{opt:dual}.
Consider taking the same steps as the proof of \cref{thm:spoa} but with the equilibrium constraint defined by the utility rule $u$.
This will result in the same linear program as in \eqref{opt:dual}, but now with the constraint set
\begin{multline}
0 \geq w(o\hs +\hs x)-\rho w(e\hs +\hs x) +\\
\sum_{\zeta \in [k]} \hs \nu_\zeta \hs \left( \hspace{-2pt} \hs \binom{n}{\zeta} u(e\hs +\hs x) \hs -\hs  \hs \sum_{\substack{0\leq\alpha\leq e\\ 0\leq\beta\leq o\\ \alpha+\beta \leq \zeta }} \hs \hspace{-2pt} {e\choose \alpha} \hs {o \choose \beta} \hs \binom{n\hs - \hs e\hs - \hs o}{\zeta \hs - \hs \alpha \hs - \hs \beta} u (e\hs +\hs x\hs +\hs \beta\hs -\hs \alpha)  \hs \hspace{-2pt} \right)\\
\forall (e,x,o) \in \mathcal{I}.
\end{multline}
At this point, the new linear program will provide tight bounds on a specified utility rule $u$.

Finally, we substitute the new decision variable $u_\zeta \in \mathbb{R}^n_{\geq 0}$ into each occurrence of $\nu_\zeta u$.
This enlarges the feasible set, which now subsumes all the feasible points that would evaluate a utility rule $u$ by satisfying $u = u_\zeta$ for all $\zeta \in [k]$.
As we do not enforce this constraint, the value of the final program \eqref{opt:util_opt} provides a lower bound on the original program, or its reciprocal provides an upper bound on the \kstrong price of anarchy under the optimal utility design.
\hfill\qed

\end{document}